\documentclass[10pt]{article}   
\usepackage{authblk}                  
\usepackage{hyperref}

\usepackage{rotating}
\usepackage{floatpag}
\usepackage{xifthen}
\usepackage{bm}

\rotfloatpagestyle{empty}
\renewcommand{\theenumi}{\arabic{enumi}}
\ifthenelse{\isundefined{\showoptional}}{
  \newcommand{\showoptional}{1}
}{}

\ifthenelse{\isundefined{\ismain}}{
  \newcommand{\ismain}{0}
}{}

\ifthenelse{\isundefined{\lecturenotes}}{
  \newcommand{\lecturenotes}{0}
}{}

\usepackage{hyperref}
\font\omding=omding

\usepackage{amsmath,amsthm,amssymb,mathabx}

\usepackage{xspace,enumerate,color,epsfig} 
\usepackage{graphicx}
\graphicspath{{.}{./figures/}}

\usepackage{tikzfig}
\usepackage{stmaryrd}
\usepackage{docmute}
\usepackage{keycommand}

\newcommand{\smallsum}{\textstyle{\sum\limits_i}\xspace}

\newcommand{\emptydiag}{\,\tikz{\node[style=empty diagram] (x) {};}\,}
\newcommand{\scalar}[1]{\,\tikz{\node[style=scalar] (x) {$#1$};}\,}
\newcommand{\dscalar}[1]{\,\tikz{\node[style=scalar, doubled] (x) {$#1$};}\,}
\newkeycommand{\boxmap}[style=small box][1]{\,\tikz{\node[style=\commandkey{style}] (x) {$#1$};\draw(0,-1.25)--(x)--(0,1.25);}\,}
\newkeycommand{\smallboxmap}[style=small box][1]{\,\tikz{\node[style=\commandkey{style}] (x) {$#1$};\draw(0,-1.25)--(x)--(0,1.25);}\,}
\newkeycommand{\shortboxmap}[style=small box][1]{\,\tikz{\node[style=\commandkey{style}] (x) {$#1$};\draw(0,-1)--(x)--(0,1);}\,}
\newkeycommand{\dboxmap}[style=dbox][1]{\,\tikz{\node[style=\commandkey{style}] (x) {$#1$};\draw[boldedge] (0,-1.25)--(x)--(0,1.25);}\,}
\newkeycommand{\shortdboxmap}[style=dbox][1]{\,\tikz{\node[style=\commandkey{style}] (x) {$#1$};\draw[boldedge] (0,-1)--(x)--(0,1);}\,}

\newcommand{\dmap}[1]{\,\tikz{\node[style=dmap] (x) {$#1$};\draw[boldedge] (0,-1.3)--(x)--(0,1.3);}\,}

\newcommand{\map}[1]{\,\tikz{\node[style=map] (x) {$#1$};\draw(0,-1.3)--(x)--(0,1.3);}\,}

\newkeycommand{\maptypes}[style=map][3]{\begin{tikzpicture}
  \begin{pgfonlayer}{nodelayer}
    \node [style=\commandkey{style}] (0) at (0, 0) {$#1$};
    \node [style=none] (1) at (0, -1.5) {};
    \node [style=none] (2) at (0, 1.5) {};
    \node [style=right label] (3) at (0.25, -1.25) {$#2$};
    \node [style=right label] (4) at (0.25, 1) {$#3$};
  \end{pgfonlayer}
  \begin{pgfonlayer}{edgelayer}
    \draw (1.center) to (0);
    \draw (0) to (2.center);
  \end{pgfonlayer}
\end{tikzpicture}}

\newcommand{\maptrans}[1]{\,\tikz{\node[style=maptrans] (x) {$#1$};\draw(0,-1.3)--(x)--(0,1.3);}\,}

\newcommand{\mapconj}[1]{\,\tikz{\node[style=mapconj] (x) {$#1$};\draw(0,-1.3)--(x)--(0,1.3);}\,}

\newkeycommand{\wirelabel}[style=white label]{\,\tikz{\node[style=\commandkey{style}]{};}\,\xspace}

\newkeycommand{\pointmap}[style=point][1]{\,\tikz{\node[style=\commandkey{style}] (x) at (0,-0.3) {$#1$};\node [style=none] (3) at (0, 0.9) {};\node [style=none] (3) at (0, -0.9) {};\draw (x)--(0,0.7);}\,}

\newkeycommand{\point}[style=point,estyle=][1]{\,\tikz{\node[style=\commandkey{style}] (x) at (0,-0.05) {$#1$};\draw [\commandkey{estyle}] (x)--(0,1.05);}\,}

\newkeycommand{\copointmap}[style=copoint][1]{\,\tikz{\node[style=\commandkey{style}] (x) at (0,0.3) {$#1$};\draw (x)--(0,-0.7);}\,}

\newkeycommand{\dpoint}[style=point][1]{\,\tikz{\node[style=\commandkey{style},doubled] (x) at (0,-0.3) {$#1$};\draw[boldedge] (x)--(0,0.7);}\,}

\newkeycommand{\kpoint}[style=kpoint,estyle=][1]{\,\tikz{\node[style=\commandkey{style}] (x) at (0,-0.05) {$#1$};\draw [\commandkey{estyle}] (x)--(0,1.05);}\,}

\newkeycommand{\kpointgray}[style=gray kpoint,estyle=][1]{\,\tikz{\node[style=\commandkey{style}] (x) at (0,-0.05) {$#1$};\draw [\commandkey{estyle}] (x)--(0,1.05);}\,}

\newcommand{\kpointconj}[1]{\,\tikz{\node[style=kpoint conjugate] (x) at (0,-0.05) {$#1$};\draw (x)--(0,1.05);}\,}

\newkeycommand{\typedkpoint}[style=kpoint,edgestyle=][2]{%
\,\begin{tikzpicture}
  \begin{pgfonlayer}{nodelayer}
    \node [style=none] (0) at (0, 1) {};
    \node [style=\commandkey{style}] (1) at (0, -0.75) {$#2$};
    \node [style=right label] (2) at (0.25, 0.5) {$#1$};
  \end{pgfonlayer}
  \begin{pgfonlayer}{edgelayer}
    \draw [\commandkey{edgestyle}] (1) to (0.center);
  \end{pgfonlayer}
\end{tikzpicture}\,}

\newkeycommand{\typedkpointdag}[style=kpoint adjoint,edgestyle=][2]{%
\,\begin{tikzpicture}
  \begin{pgfonlayer}{nodelayer}
    \node [style=none] (0) at (0, -1) {};
    \node [style=\commandkey{style}] (1) at (0, 0.75) {$#2$};
    \node [style=right label] (2) at (0.25, -0.5) {$#1$};
  \end{pgfonlayer}
  \begin{pgfonlayer}{edgelayer}
    \draw [\commandkey{edgestyle}] (0.center) to (1);
  \end{pgfonlayer}
\end{tikzpicture}\,}

\newkeycommand{\bistate}[style=kpoint][1]{\,\begin{tikzpicture}
  \begin{pgfonlayer}{nodelayer}
    \node [style=\commandkey{style}, minimum width=1 cm, inner sep=2pt] (0) at (0, -0.25) {$#1$};
    \node [style=none] (1) at (-0.75, 0) {};
    \node [style=none] (2) at (0.75, 0) {};
    \node [style=none] (3) at (-0.75, 0.88) {};
    \node [style=none] (4) at (0.75, 0.88) {};
  \end{pgfonlayer}
  \begin{pgfonlayer}{edgelayer}
    \draw (1.center) to (3.center);
    \draw (2.center) to (4.center);
  \end{pgfonlayer}
\end{tikzpicture}\,}

\newkeycommand{\bistatebig}[style=kpoint][1]{\,\begin{tikzpicture}
  \begin{pgfonlayer}{nodelayer}
    \node [style=\commandkey{style}, minimum width=, minimum width=1.26 cm, inner sep=2pt] (0) at (0, -0.25) {$#1$};
    \node [style=none] (1) at (-1, 0) {};
    \node [style=none] (2) at (1, 0) {};
    \node [style=none] (3) at (-1, 0.88) {};
    \node [style=none] (4) at (1, 0.88) {};
  \end{pgfonlayer}
  \begin{pgfonlayer}{edgelayer}
    \draw (1.center) to (3.center);
    \draw (2.center) to (4.center);
  \end{pgfonlayer}
\end{tikzpicture}\,}

\newkeycommand{\dbistate}[style=dkpoint][1]{\,\begin{tikzpicture}
  \begin{pgfonlayer}{nodelayer}
    \node [style=\commandkey{style}, minimum width=1 cm, inner sep=2pt] (0) at (0, -0.25) {$#1$};
    \node [style=none] (1) at (-0.75, 0) {};
    \node [style=none] (2) at (0.75, 0) {};
    \node [style=none] (3) at (-0.75, 1.0) {};
    \node [style=none] (4) at (0.75, 1.0) {};
  \end{pgfonlayer}
  \begin{pgfonlayer}{edgelayer}
    \draw [boldedge] (1.center) to (3.center);
    \draw [boldedge] (2.center) to (4.center);
  \end{pgfonlayer}
\end{tikzpicture}\,}

\newkeycommand{\bistatebraket}[style1=kpoint,style2=kpointadj][2]{\,\begin{tikzpicture}
  \begin{pgfonlayer}{nodelayer}
    \node [style=\commandkey{style1}, minimum width=1 cm, inner sep=2pt] (0) at (0, -0.75) {$#1$};
    \node [style=none] (1) at (-0.75, -0.5) {};
    \node [style=none] (2) at (0.75, -0.5) {};
    \node [style=none] (3) at (-0.75, 0.5) {};
    \node [style=none] (4) at (0.75, 0.5) {};
    \node [style=\commandkey{style2}, minimum width=1 cm, inner sep=2pt] (5) at (0, 0.75) {$#2$};
  \end{pgfonlayer}
  \begin{pgfonlayer}{edgelayer}
    \draw (1.center) to (3.center);
    \draw (2.center) to (4.center);
  \end{pgfonlayer}
\end{tikzpicture}\,}

\newkeycommand{\weight}[style=dkpoint][1]{\,\begin{tikzpicture}
  \begin{pgfonlayer}{nodelayer}
    \node [style=\commandkey{style}] (0) at (0, -0.5) {$#1$};
    \node [style=upground] (1) at (0, 0.75) {};
  \end{pgfonlayer}
  \begin{pgfonlayer}{edgelayer}
    \draw [style=boldedge] (0) to (1);
  \end{pgfonlayer}
\end{tikzpicture}\,}

\newkeycommand{\dkpoint}[style=dkpoint][1]{\,\tikz{\node[style=\commandkey{style}] (x) at (0,-0.1) {$#1$};\draw[boldedge] (x)--(0,1);}\,}
\newkeycommand{\dkpointadj}[style=dkpointadj][1]{\,\tikz{\node[style=\commandkey{style}] (x) at (0,0.1) {$#1$};\draw[boldedge] (x)--(0,-1);}\,}
\newkeycommand{\dkpointtrans}[style=dkpointtrans][1]{\,\tikz{\node[style=\commandkey{style}] (x) at (0,0.1) {$#1$};\draw[boldedge] (x)--(0,-1);}\,}
\newkeycommand{\dkpointconj}[style=dkpointconj][1]{\,\tikz{\node[style=\commandkey{style}] (x) at (0,-0.1) {$#1$};\draw[boldedge] (x)--(0,1);}\,}

\newkeycommand{\blackkpoint}[style=black kpoint][1]{\,\tikz{\node[style=\commandkey{style}] (x) at (0,-0.1) {$#1$};\draw (x)--(0,1);}\,}
\newkeycommand{\blackkpointadj}[style=black kpointadj][1]{\,\tikz{\node[style=\commandkey{style}] (x) at (0,0.1) {$#1$};\draw (x)--(0,-1);}\,}
\newkeycommand{\blackdkpoint}[style=black dkpoint][1]{\,\tikz{\node[style=\commandkey{style}] (x) at (0,-0.1) {$#1$};\draw[boldedge] (x)--(0,1);}\,}
\newkeycommand{\blackdkpointadj}[style=black dkpointadj][1]{\,\tikz{\node[style=\commandkey{style}] (x) at (0,0.1) {$#1$};\draw[boldedge] (x)--(0,-1);}\,}

\newcommand{\trace}{\,\begin{tikzpicture}
  \begin{pgfonlayer}{nodelayer}
    \node [style=none] (0) at (0, -0.60) {};
    \node [style=upground] (1) at (0, 0.40) {};
  \end{pgfonlayer}
  \begin{pgfonlayer}{edgelayer}
    \draw [style=boldedge] (0.center) to (1);
  \end{pgfonlayer}
\end{tikzpicture}\,}

\newcommand\discard\trace

\newkeycommand{\pointketbra}[style1=copoint,style2=point,estyle=][2]{\,%
\begin{tikzpicture}
  \begin{pgfonlayer}{nodelayer}
    \node [style=none] (0) at (0, -1.75) {};
    \node [style=\commandkey{style1}] (1) at (0, -0.75) {$#1$};
    \node [style=\commandkey{style2}] (2) at (0, 0.75) {$#2$};
    \node [style=none] (3) at (0, 1.75) {};
  \end{pgfonlayer}
  \begin{pgfonlayer}{edgelayer}
    \draw [\commandkey{estyle}] (0.center) to (1);
    \draw [\commandkey{estyle}] (2) to (3.center);
  \end{pgfonlayer}
\end{tikzpicture}\,}

\newkeycommand{\twocopointketbra}[style1=copoint,style2=copoint,style3=point][3]{\,%
\begin{tikzpicture}
  \begin{pgfonlayer}{nodelayer}
    \node [style=none] (0) at (-0.75, -1.75) {};
    \node [style=none] (1) at (0.75, -1.75) {};
    \node [style=\commandkey{style1}] (2) at (-0.75, -0.75) {$#1$};
    \node [style=\commandkey{style2}] (3) at (0.75, -0.75) {$#2$};
    \node [style=\commandkey{style3}] (4) at (0, 0.75) {$#3$};
    \node [style=none] (5) at (0, 1.75) {};
  \end{pgfonlayer}
  \begin{pgfonlayer}{edgelayer}
    \draw (0.center) to (2);
    \draw (1.center) to (3);
    \draw (4) to (5.center);
  \end{pgfonlayer}
\end{tikzpicture}\,}

\newkeycommand{\twopointketbra}[style1=copoint,style2=point,style3=point][3]{\,%
\begin{tikzpicture}
  \begin{pgfonlayer}{nodelayer}
    \node [style=none] (0) at (0, -1.75) {};
    \node [style=none] (1) at (-0.75, 1.75) {};
    \node [style=\commandkey{style1}] (2) at (0, -0.75) {$#1$};
    \node [style=\commandkey{style2}] (3) at (-0.75, 0.75) {$#2$};
    \node [style=\commandkey{style3}] (4) at (0.75, 0.75) {$#3$};
    \node [style=none] (5) at (0.75, 1.75) {};
  \end{pgfonlayer}
  \begin{pgfonlayer}{edgelayer}
    \draw (0.center) to (2);
    \draw (1.center) to (3);
    \draw (4) to (5.center);
  \end{pgfonlayer}
\end{tikzpicture}\,}

\newkeycommand{\pointbraket}[style1=point,style2=copoint,estyle=][2]{\,%
\begin{tikzpicture}
  \begin{pgfonlayer}{nodelayer}
    \node [style=\commandkey{style1}] (0) at (0, -0.625) {$#1$};
    \node [style=\commandkey{style2}] (1) at (0, 0.625) {$#2$};
  \end{pgfonlayer}
  \begin{pgfonlayer}{edgelayer}
    \draw [\commandkey{estyle}] (0) to (1);
  \end{pgfonlayer}
\end{tikzpicture}\,}

\newkeycommand{\kpointketbra}[style1=kpointdag,style2=kpoint,estyle=][2]{\,%
\begin{tikzpicture}
  \begin{pgfonlayer}{nodelayer}
    \node [style=none] (0) at (0, -1.9) {};
    \node [style=\commandkey{style1}] (1) at (0, -0.95) {$#1$};
    \node [style=\commandkey{style2}] (2) at (0, 0.95) {$#2$};
    \node [style=none] (3) at (0, 1.9) {};
  \end{pgfonlayer}
  \begin{pgfonlayer}{edgelayer}
    \draw [\commandkey{estyle}] (0.center) to (1);
    \draw [\commandkey{estyle}] (2) to (3.center);
  \end{pgfonlayer}
\end{tikzpicture}\,}

\newkeycommand{\kpointmap}[style1=kpoint,style2=map][2]{\,%
\begin{tikzpicture}
  \begin{pgfonlayer}{nodelayer}
    \node [style=\commandkey{style1}] (0) at (0, -1.1) {$#1$};
    \node [style=\commandkey{style2}] (1) at (0, 0.5) {$#2$};
    \node [style=none] (2) at (0, 1.75) {};
  \end{pgfonlayer}
  \begin{pgfonlayer}{edgelayer}
    \draw (0) to (1);
    \draw (1) to (2.center);
  \end{pgfonlayer}
\end{tikzpicture}%
\,}

\newkeycommand{\kpointmaptrans}[style1=kpointconj,style2=mapadj][2]{\,%
\begin{tikzpicture}
  \begin{pgfonlayer}{nodelayer}
    \node [style=\commandkey{style1}] (0) at (0, -1.1) {$#1$};
    \node [style=\commandkey{style2}] (1) at (0, 0.5) {$#2$};
    \node [style=none] (2) at (0, 1.75) {};
  \end{pgfonlayer}
  \begin{pgfonlayer}{edgelayer}
    \draw (0) to (1);
    \draw (1) to (2.center);
  \end{pgfonlayer}
\end{tikzpicture}%
\,}

\newkeycommand{\kpointmapadj}[style1=kpointadj,style2=map][2]{\,%
\begin{tikzpicture}
  \begin{pgfonlayer}{nodelayer}
    \node [style=\commandkey{style1}] (0) at (0, 1.1) {$#1$};
    \node [style=\commandkey{style2}] (1) at (0, -0.5) {$#2$};
    \node [style=none] (2) at (0, -1.75) {};
  \end{pgfonlayer}
  \begin{pgfonlayer}{edgelayer}
    \draw (0) to (1);
    \draw (1) to (2.center);
  \end{pgfonlayer}
\end{tikzpicture}%
\,}

\newkeycommand{\dkpointmap}[style1=dkpoint,style2=dmap][2]{\,%
\begin{tikzpicture}
  \begin{pgfonlayer}{nodelayer}
    \node [style=\commandkey{style1}] (0) at (0, -1.2) {$#1$};
    \node [style=\commandkey{style2}] (1) at (0, 0.4) {$#2$};
    \node [style=none] (2) at (0, 1.7) {};
  \end{pgfonlayer}
  \begin{pgfonlayer}{edgelayer}
    \draw[style=boldedge] (0) to (1);
    \draw[style=boldedge] (1) to (2.center);
  \end{pgfonlayer}
\end{tikzpicture}%
\,}

\newkeycommand{\dkpointmapx}[style1=dkpoint,style2=dmap][2]{\,%
\begin{tikzpicture}
  \begin{pgfonlayer}{nodelayer}
    \node [style=\commandkey{style1}] (0) at (0, -1.4) {$#1$};
    \node [style=\commandkey{style2}] (1) at (0, 0.6) {$#2$};
    \node [style=none] (2) at (0, 1.9) {};
  \end{pgfonlayer}
  \begin{pgfonlayer}{edgelayer}
    \draw[style=boldedge] (0) to (1);
    \draw[style=boldedge] (1) to (2.center);
  \end{pgfonlayer}
\end{tikzpicture}%
\,}

\newkeycommand{\boxcopointmapdag}[style1=map,style2=copoint][2]{\,%
\begin{tikzpicture}
  \begin{pgfonlayer}{nodelayer}
    \node [style=none] (0) at (0, -1.75) {};
    \node [style=\commandkey{style1}] (1) at (0, -0.5) {$#1$};
    \node [style=\commandkey{style2}] (2) at (0, 1) {$#2$};
  \end{pgfonlayer}
  \begin{pgfonlayer}{edgelayer}
    \draw (0.center) to (1);
    \draw (1) to (2);
  \end{pgfonlayer}
\end{tikzpicture}%
\,}

\newkeycommand{\kpointdagmap}[style1=map,style2=kpointdag][2]{\,%
\begin{tikzpicture}
  \begin{pgfonlayer}{nodelayer}
    \node [style=none] (0) at (0, -1.75) {};
    \node [style=\commandkey{style1}] (1) at (0, -0.5) {$#1$};
    \node [style=\commandkey{style2}] (2) at (0, 1.1) {$#2$};
  \end{pgfonlayer}
  \begin{pgfonlayer}{edgelayer}
    \draw (0.center) to (1);
    \draw (1) to (2);
  \end{pgfonlayer}
\end{tikzpicture}%
\,}

\newkeycommand{\sandwichmap}[style1=point,style2=map,style3=copoint][3]{\,%
\begin{tikzpicture}
  \begin{pgfonlayer}{nodelayer}
    \node [style=\commandkey{style1}] (0) at (0, -1.50) {$#1$};
    \node [style=\commandkey{style2}] (1) at (0, 0) {$#2$};
    \node [style=\commandkey{style3}] (2) at (0, 1.50) {$#3$};
  \end{pgfonlayer}
  \begin{pgfonlayer}{edgelayer}
    \draw (0) to (1);
    \draw (1) to (2);
  \end{pgfonlayer}
\end{tikzpicture}%
}

\newkeycommand{\kpointsandwichmap}[style1=kpoint,style2=map,style3=kpointdag][3]{\,%
\begin{tikzpicture}
  \begin{pgfonlayer}{nodelayer}
    \node [style=\commandkey{style1}] (0) at (0, -1.5) {$#1$};
    \node [style=\commandkey{style2}] (1) at (0, 0) {$#2$};
    \node [style=\commandkey{style3}] (2) at (0, 1.5) {$#3$};
  \end{pgfonlayer}
  \begin{pgfonlayer}{edgelayer}
    \draw (0) to (1);
    \draw (1) to (2);
  \end{pgfonlayer}
\end{tikzpicture}%
}

\newkeycommand{\sandwichmapconj}[style1=point,style2=mapconj,style3=copoint][3]{\,%
\begin{tikzpicture}
  \begin{pgfonlayer}{nodelayer}
    \node [style=\commandkey{style1}] (0) at (0, -1.50) {$#1$};
    \node [style=\commandkey{style2}] (1) at (0, 0) {$#2$};
    \node [style=\commandkey{style3}] (2) at (0, 1.50) {$#3$};
  \end{pgfonlayer}
  \begin{pgfonlayer}{edgelayer}
    \draw (0) to (1);
    \draw (1) to (2);
  \end{pgfonlayer}
\end{tikzpicture}%
}

\newkeycommand{\sandwichtwo}[style1=point,style2=map,style3=map,style4=copoint][4]{\,%
\begin{tikzpicture}
  \begin{pgfonlayer}{nodelayer}
    \node [style=\commandkey{style2}] (0) at (0, -1) {$#2$};
    \node [style=\commandkey{style3}] (1) at (0, 1) {$#3$};
    \node [style=\commandkey{style1}] (2) at (0, -2.6) {$#1$};
    \node [style=\commandkey{style4}] (3) at (0, 2.6) {$#4$};
  \end{pgfonlayer}
  \begin{pgfonlayer}{edgelayer}
    \draw (2) to (0);
    \draw (0) to (1);
    \draw (1) to (3);
  \end{pgfonlayer}
\end{tikzpicture}\,}

\newkeycommand{\longbraket}[style1=point,style2=copoint][2]{\,%
\begin{tikzpicture}
  \begin{pgfonlayer}{nodelayer}
    \node [style=\commandkey{style1}] (0) at (0, -2.6) {$#1$};
    \node [style=\commandkey{style2}] (1) at (0, 2.6) {$#2$};
  \end{pgfonlayer}
  \begin{pgfonlayer}{edgelayer}
    \draw (0) to (1);
  \end{pgfonlayer}
\end{tikzpicture}\,}

\newkeycommand{\onb}[style=point]{\ensuremath{\{\,\tikz{\node[style=\commandkey{style}] (x) at (0,-0.3) {$j$};\draw (x)--(0,0.7);}\,\}}\xspace}

\newkeycommand{\phasebraket}[style1=gray phase dot,style2=white phase dot][2]{\begin{tikzpicture}
  \begin{pgfonlayer}{nodelayer}
    \node [style=\commandkey{style1}] (0) at (0, -0.75) {$#1$};
    \node [style=\commandkey{style2}] (1) at (0, 0.5) {$#2$};
  \end{pgfonlayer}
  \begin{pgfonlayer}{edgelayer}
    \draw (1) to (0);
  \end{pgfonlayer}
\end{tikzpicture}}

\newkeycommand{\phase}[style=white phase dot][1]{\,\begin{tikzpicture}
    \begin{pgfonlayer}{nodelayer}
        \node [style=none] (0) at (0, 1) {};
        \node [style=\commandkey{style}] (2) at (0, -0) {$#1$}; 
        \node [style=none] (3) at (0, -1) {};
    \end{pgfonlayer}
    \begin{pgfonlayer}{edgelayer}
        \draw (2) to (0.center);
        \draw (3.center) to (2);
    \end{pgfonlayer}
\end{tikzpicture}\,}

\newkeycommand{\dphase}[style=white phase ddot][1]{\,\begin{tikzpicture}
    \begin{pgfonlayer}{nodelayer}
        \node [style=none] (0) at (0, 1) {};
        \node [style=\commandkey{style}] (2) at (0, -0) {$#1$};
        \node [style=none] (3) at (0, -1) {};
    \end{pgfonlayer}
    \begin{pgfonlayer}{edgelayer}
        \draw [boldedge] (2) to (0.center);
        \draw [boldedge] (3.center) to (2);
    \end{pgfonlayer}
\end{tikzpicture}\,}

\newkeycommand{\dphasepoint}[style=white phase ddot][1]{\,\begin{tikzpicture}[yshift=-3mm]
  \begin{pgfonlayer}{nodelayer}
    \node [style=none] (0) at (0, 1) {};
    \node [style=\commandkey{style}] (1) at (0, 0) {$#1$};
  \end{pgfonlayer}
  \begin{pgfonlayer}{edgelayer}
    \draw [style=boldedge] (0.center) to (1);
  \end{pgfonlayer}
\end{tikzpicture}\,}

\newkeycommand{\dphasecopoint}[style=white phase ddot][1]{\,\begin{tikzpicture}[yshift=3mm,yscale=-1]
  \begin{pgfonlayer}{nodelayer}
    \node [style=none] (0) at (0, 1) {};
    \node [style=\commandkey{style}] (1) at (0, 0) {$#1$};
  \end{pgfonlayer}
  \begin{pgfonlayer}{edgelayer}
    \draw [style=boldedge] (0.center) to (1);
  \end{pgfonlayer}
\end{tikzpicture}\,}

\newkeycommand{\dphasepointsm}[style=white phase ddot][1]{\,\begin{tikzpicture}[yshift=-3mm]
  \begin{pgfonlayer}{nodelayer}
    \node [style=none] (0) at (0, 1) {};
    \node [style=\commandkey{style}] (1) at (0, 0) {\footnotesize\!$#1$\!};
  \end{pgfonlayer}
  \begin{pgfonlayer}{edgelayer}
    \draw [style=boldedge] (0.center) to (1);
  \end{pgfonlayer}
\end{tikzpicture}\,}

\newkeycommand{\phasepoint}[style=white phase dot][1]{\,\begin{tikzpicture}
  \begin{pgfonlayer}{nodelayer}
    \node [style=none] (0) at (0, 0.75) {};
    \node [style=\commandkey{style}] (1) at (0, -0.25) {$#1$};
  \end{pgfonlayer}
  \begin{pgfonlayer}{edgelayer}
    \draw (0.center) to (1);
  \end{pgfonlayer}
\end{tikzpicture}\,}

\newkeycommand{\phasecopoint}[style=white phase dot][1]{\,\begin{tikzpicture}[yshift=3mm]
  \begin{pgfonlayer}{nodelayer}
    \node [style=none] (0) at (0, -0.75) {};
    \node [style=\commandkey{style}] (1) at (0, 0.25) {$#1$};
  \end{pgfonlayer}
  \begin{pgfonlayer}{edgelayer}
    \draw (0.center) to (1);
  \end{pgfonlayer}
\end{tikzpicture}\,}

\newkeycommand{\kpointbraket}[style1=kpoint,style2=kpoint adjoint,estyle=][2]{\,%
\begin{tikzpicture}
  \begin{pgfonlayer}{nodelayer}
    \node [style=\commandkey{style1}] (0) at (0, -0.735) {$#1$};
    \node [style=\commandkey{style2}] (1) at (0, 0.735) {$#2$};
  \end{pgfonlayer}
  \begin{pgfonlayer}{edgelayer}
    \draw [\commandkey{estyle}] (0) to (1);
  \end{pgfonlayer}
\end{tikzpicture}\,}

\newkeycommand{\kkpointbraket}[style1=kpoint,style2=kpoint adjoint,estyle=][2]{\,%
\begin{tikzpicture}
  \begin{pgfonlayer}{nodelayer}
    \node [style=\commandkey{style1}] (0) at (0, -0.685) {$#1$};
    \node [style=\commandkey{style2}] (1) at (0, 0.685) {$#2$};
  \end{pgfonlayer}
  \begin{pgfonlayer}{edgelayer}
    \draw [\commandkey{estyle}] (0) to (1);
  \end{pgfonlayer}
\end{tikzpicture}\,}

\newkeycommand{\kpointbraketx}[style1=kpoint,style2=kpoint adjoint,estyle=][2]{\,%
\begin{tikzpicture}
  \begin{pgfonlayer}{nodelayer}
    \node [style=\commandkey{style1}] (0) at (0, -0.885) {$#1$};
    \node [style=\commandkey{style2}] (1) at (0, 0.885) {$#2$};
  \end{pgfonlayer}
  \begin{pgfonlayer}{edgelayer}
    \draw [\commandkey{estyle}] (0) to (1);
  \end{pgfonlayer}
\end{tikzpicture}\,}

\newcommand{\bra}[1]{\ensuremath{\left\langle #1 \right|}}
\newcommand{\ket}[1]{\ensuremath{\left|  #1 \right\rangle}}

\newcommand{\braket}[2]{\ensuremath{\langle#1|#2\rangle}}

\tikzstyle{cdiag}=[matrix of math nodes, row sep=3em, column sep=3em, text height=1.5ex, text depth=0.25ex,inner sep=0.5em]
\tikzstyle{arrow above}=[transform canvas={yshift=0.5ex}]
\tikzstyle{arrow below}=[transform canvas={yshift=-0.5ex}]

\tikzset{
  col/.cd,
  0/.style={draw,fill=white}
  1/.style={draw,fill=blue}
}

\tikzstyle{B}=[draw,fill=gray!90!black]
\tikzstyle{W}=[draw,fill=white]
\tikzstyle{G}=[draw,fill=gray!50!white]

\tikzstyle{env}=[copoint,regular polygon rotate=0,minimum width=0.2cm, fill=black]

\tikzstyle{probs}=[shape=semicircle,fill=white,draw=black,shape border rotate=180,minimum width=1.2cm]

\tikzstyle{every picture}=[baseline=-0.25em,scale=0.5]
\tikzstyle{dotpic}=[] 
\tikzstyle{diredges}=[every to/.style={diredge}]
\tikzstyle{math matrix}=[matrix of math nodes,left delimiter=(,right delimiter=),inner sep=2pt,column sep=1em,row sep=0.5em,nodes={inner sep=0pt},text height=1.5ex, text depth=0.25ex]

\tikzstyle{inline text}=[text height=1.5ex, text depth=0.25ex,yshift=0.5mm]
\tikzstyle{label}=[font=\footnotesize,text height=1.5ex, text depth=0.25ex,yshift=0.5mm]
\tikzstyle{left label}=[label,anchor=east,xshift=1.5mm]
\tikzstyle{right label}=[label,anchor=west,xshift=-1.5mm]

\tikzstyle{braceedge}=[decorate,decoration={brace,amplitude=2mm,raise=-1mm}]
\tikzstyle{small braceedge}=[decorate,decoration={brace,amplitude=1mm,raise=-1mm}]

\tikzstyle{doubled}=[line width=1.6pt] 
\tikzstyle{boldedge}=[doubled,shorten <=-0.17mm,shorten >=-0.17mm]
\tikzstyle{boldedgegray}=[doubled,gray,shorten <=-0.17mm,shorten >=-0.17mm]
\tikzstyle{singleedgegray}=[gray]

\tikzstyle{semidoubled}=[line width=1.4pt] 
\tikzstyle{semiboldedgegray}=[semidoubled,gray,shorten <=-0.17mm,shorten >=-0.17mm]

\tikzstyle{boxedge}=[semiboldedgegray]

\tikzstyle{boldedgedashed}=[very thick,dashed,shorten <=-0.17mm,shorten >=-0.17mm]
\tikzstyle{vboldedgedashed}=[doubled,dashed,shorten <=-0.17mm,shorten >=-0.17mm]
\tikzstyle{left hook arrow}=[left hook-latex]
\tikzstyle{right hook arrow}=[right hook-latex]
\tikzstyle{sembracket}=[line width=0.5pt,shorten <=-0.07mm,shorten >=-0.07mm]

\tikzstyle{causal edge}=[->,thick,gray]
\tikzstyle{causal nondir}=[thick,gray]
\tikzstyle{timeline}=[thick,gray, dashed]

\tikzstyle{cedge}=[<->,thick,gray!70!white]

\tikzstyle{empty diagram}=[draw=gray!40!white,dashed,shape=rectangle,minimum width=1cm,minimum height=1cm]
\tikzstyle{empty diagram small}=[draw=gray!50!white,dashed,shape=rectangle,minimum width=0.6cm,minimum height=0.5cm]

\tikzstyle{dot}=[inner sep=0mm,minimum width=2mm,minimum height=2mm,draw,shape=circle]  
\tikzstyle{Wsquare}=[white dot, shape=regular polygon, rounded corners=0.8 mm, minimum size=3.3 mm, regular polygon sides=3, outer sep=-0.2mm]
\tikzstyle{Wsquareadj}=[white dot, shape=regular polygon, rounded corners=0.8 mm, minimum size=3.3 mm, regular polygon sides=3, outer sep=-0.2mm, regular polygon rotate=180]
\tikzstyle{ddot}=[inner sep=0mm, doubled, minimum width=2.5mm,minimum height=2.5mm,draw,shape=circle]

\tikzstyle{black dot}=[dot,fill=black]
\tikzstyle{white dot}=[dot,fill=white,,text depth=-0.2mm]
\tikzstyle{white Wsquare}=[Wsquare,fill=gray,,text depth=-0.2mm]
\tikzstyle{white Wsquareadj}=[Wsquareadj,fill=white,,text depth=-0.2mm]
\tikzstyle{green dot}=[white dot] 
\tikzstyle{gray dot}=[dot,fill=gray!40!white,,text depth=-0.2mm]
\tikzstyle{red dot}=[gray dot] 

\tikzstyle{black ddot}=[ddot,fill=black]
\tikzstyle{white ddot}=[ddot,fill=white]
\tikzstyle{gray ddot}=[ddot,fill=gray!40!white]

\tikzstyle{gray edge}=[gray!60!white]

\tikzstyle{small dot}=[inner sep=0.5mm,minimum width=0pt,minimum height=0pt,draw,shape=circle]

\tikzstyle{small black dot}=[small dot,fill=black]
\tikzstyle{small white dot}=[small dot,fill=white]
\tikzstyle{small gray dot}=[small dot,fill=gray!40!white]

\tikzstyle{causal dot}=[inner sep=0.4mm,minimum width=0pt,minimum height=0pt,draw=white,shape=circle,fill=gray!40!white]

\tikzstyle{phase dimensions}=[minimum size=5mm,font=\footnotesize,rectangle,rounded corners=2.5mm,inner sep=0.2mm,outer sep=-2mm]
\tikzstyle{dphase dimensions}=[minimum size=5mm,font=\footnotesize,rectangle,rounded corners=2.5mm,inner sep=0.2mm,outer sep=-2mm]

\tikzstyle{white phase dot}=[dot,fill=white,phase dimensions]
\tikzstyle{white phase ddot}=[ddot,fill=white,dphase dimensions]

\tikzstyle{white rect ddot}=[draw=black,fill=white,doubled,minimum size=5mm,font=\footnotesize,rectangle,rounded corners=2.5mm,inner sep=0.2mm]
\tikzstyle{gray rect ddot}=[draw=black,fill=gray!40!white,doubled,minimum size=6mm,font=\footnotesize,rectangle,rounded corners=3mm]

\tikzstyle{gray phase dot}=[dot,fill=gray!40!white,phase dimensions]
\tikzstyle{gray phase ddot}=[ddot,fill=gray!40!white,dphase dimensions]
\tikzstyle{grey phase dot}=[gray phase dot]
\tikzstyle{grey phase ddot}=[gray phase ddot]

\tikzstyle{small phase dimensions}=[minimum size=4mm,font=\tiny,rectangle,rounded corners=2mm,inner sep=0.2mm,outer sep=-2mm]
\tikzstyle{small dphase dimensions}=[minimum size=4mm,font=\tiny,rectangle,rounded corners=2mm,inner sep=0.2mm,outer sep=-2mm]

\tikzstyle{small gray phase dot}=[dot,fill=gray!40!white,small phase dimensions]
\tikzstyle{small gray phase ddot}=[ddot,fill=gray!40!white,small dphase dimensions]

\tikzstyle{small map}=[draw,shape=rectangle,minimum height=4mm,minimum width=4mm,fill=white]

\tikzstyle{cnot}=[fill=white,shape=circle,inner sep=-1.4pt]

\tikzstyle{asym hadamard}=[fill=white,draw,shape=NEbox,inner sep=0.6mm,font=\footnotesize,minimum height=4mm]
\tikzstyle{asym hadamard conj}=[fill=white,draw,shape=NWbox,inner sep=0.6mm,font=\footnotesize,minimum height=4mm]
\tikzstyle{asym hadamard dag}=[fill=white,draw,shape=SEbox,inner sep=0.6mm,font=\footnotesize,minimum height=4mm]

\tikzstyle{hadamard}=[fill=white,draw,inner sep=0.6mm,font=\footnotesize,minimum height=4mm,minimum width=4mm]
\tikzstyle{small hadamard}=[fill=white,draw,inner sep=0.6mm,minimum height=1.5mm,minimum width=1.5mm]
\tikzstyle{small hadamard rotate}=[small hadamard,rotate=45]
\tikzstyle{dhadamard}=[hadamard,doubled]
\tikzstyle{small dhadamard}=[small hadamard,doubled]
\tikzstyle{small dhadamard rotate}=[small hadamard rotate,doubled]
\tikzstyle{antipode}=[white dot,inner sep=0.3mm,font=\footnotesize]

\tikzstyle{scalar}=[diamond,draw,inner sep=0.5pt,font=\small]
\tikzstyle{dscalar}=[diamond,doubled, draw,inner sep=0.5pt,font=\small]

\tikzstyle{small box}=[rectangle,inline text,fill=white,draw,minimum height=5mm,yshift=-0.5mm,minimum width=5mm,font=\small]
\tikzstyle{small gray box}=[small box,fill=gray!30]
\tikzstyle{medium box}=[rectangle,inline text,fill=white,draw,minimum height=5mm,yshift=-0.5mm,minimum width=10mm,font=\small]
\tikzstyle{square box}=[small box] 
\tikzstyle{medium gray box}=[small box,fill=gray!30]
\tikzstyle{semilarge box}=[rectangle,inline text,fill=white,draw,minimum height=5mm,yshift=-0.5mm,minimum width=12.5mm,font=\small]
\tikzstyle{large box}=[rectangle,inline text,fill=white,draw,minimum height=5mm,yshift=-0.5mm,minimum width=15mm,font=\small]
\tikzstyle{large gray box}=[small box,fill=gray!30]

\tikzstyle{Bayes box}=[rectangle,fill=black,draw, minimum height=3mm, minimum width=3mm]

\tikzstyle{gray square point}=[small box,fill=gray!50]

\tikzstyle{dphase box white}=[dhadamard]
\tikzstyle{dphase box gray}=[dhadamard,fill=gray!50!white]
\tikzstyle{phase box white}=[hadamard]
\tikzstyle{phase box gray}=[hadamard,fill=gray!50!white]

\tikzstyle{point}=[regular polygon,regular polygon sides=3,draw,scale=0.75,inner sep=-0.5pt,minimum width=9mm,fill=white,regular polygon rotate=180]
\tikzstyle{point nosep}=[regular polygon,regular polygon sides=3,draw,scale=0.75,inner sep=-2pt,minimum width=9mm,fill=white,regular polygon rotate=180]
\tikzstyle{copoint}=[regular polygon,regular polygon sides=3,draw,scale=0.75,inner sep=-0.5pt,minimum width=9mm,fill=white]
\tikzstyle{dpoint}=[point,doubled]
\tikzstyle{dcopoint}=[copoint,doubled]

\tikzstyle{pointgrow}=[shape=cornerpoint,kpoint common,scale=0.75,inner sep=3pt]
\tikzstyle{pointgrow dag}=[shape=cornercopoint,kpoint common,scale=0.75,inner sep=3pt]

\tikzstyle{wide copoint}=[fill=white,draw,shape=isosceles triangle,shape border rotate=90,isosceles triangle stretches=true,inner sep=0pt,minimum width=1.5cm,minimum height=6.12mm]
\tikzstyle{wide point}=[fill=white,draw,shape=isosceles triangle,shape border rotate=-90,isosceles triangle stretches=true,inner sep=0pt,minimum width=1.5cm,minimum height=6.12mm,yshift=-0.0mm]
\tikzstyle{wide point plus}=[fill=white,draw,shape=isosceles triangle,shape border rotate=-90,isosceles triangle stretches=true,inner sep=0pt,minimum width=1.74cm,minimum height=7mm,yshift=-0.0mm]

\tikzstyle{wide dpoint}=[fill=white,doubled,draw,shape=isosceles triangle,shape border rotate=-90,isosceles triangle stretches=true,inner sep=0pt,minimum width=1.5cm,minimum height=6.12mm,yshift=-0.0mm]

\tikzstyle{tinypoint}=[regular polygon,regular polygon sides=3,draw,scale=0.55,inner sep=-0.15pt,minimum width=6mm,fill=white,regular polygon rotate=180] 

\tikzstyle{white point}=[point]
\tikzstyle{white dpoint}=[dpoint]
\tikzstyle{green point}=[white point] 
\tikzstyle{white copoint}=[copoint]
\tikzstyle{gray point}=[point,fill=gray!40!white]
\tikzstyle{gray dpoint}=[gray point,doubled]
\tikzstyle{red point}=[gray point] 
\tikzstyle{gray copoint}=[copoint,fill=gray!40!white]
\tikzstyle{gray dcopoint}=[gray copoint,doubled]

\tikzstyle{white point guide}=[regular polygon,regular polygon sides=3,font=\scriptsize,draw,scale=0.65,inner sep=-0.5pt,minimum width=9mm,fill=white,regular polygon rotate=180]

\tikzstyle{black point}=[point,fill=black,font=\color{white}]
\tikzstyle{black copoint}=[copoint,fill=black,font=\color{white}]

\tikzstyle{tiny gray point}=[tinypoint,fill=gray!40!white]

\tikzstyle{diredge}=[->]
\tikzstyle{ddiredge}=[<->]
\tikzstyle{rdiredge}=[<-]
\tikzstyle{thickdiredge}=[->, very thick]
\tikzstyle{pointer edge}=[->,very thick,gray]
\tikzstyle{pointer edge part}=[very thick,gray]
\tikzstyle{dashed edge}=[dashed]
\tikzstyle{thick dashed edge}=[very thick,dashed]
\tikzstyle{thick gray dashed edge}=[thick dashed edge,gray!40]
\tikzstyle{thick map edge}=[very thick,|->]

\makeatletter
\newcommand{\boxshape}[3]{%
\pgfdeclareshape{#1}{
\inheritsavedanchors[from=rectangle] 
\inheritanchorborder[from=rectangle]
\inheritanchor[from=rectangle]{center}
\inheritanchor[from=rectangle]{north}
\inheritanchor[from=rectangle]{south}
\inheritanchor[from=rectangle]{west}
\inheritanchor[from=rectangle]{east}
\backgroundpath{
\southwest \pgf@xa=\pgf@x \pgf@ya=\pgf@y
\northeast \pgf@xb=\pgf@x \pgf@yb=\pgf@y

\@tempdima=#2
\@tempdimb=#3

\pgfpathmoveto{\pgfpoint{\pgf@xa - 5pt + \@tempdima}{\pgf@ya}}
\pgfpathlineto{\pgfpoint{\pgf@xa - 5pt - \@tempdima}{\pgf@yb}}
\pgfpathlineto{\pgfpoint{\pgf@xb + 5pt + \@tempdimb}{\pgf@yb}}
\pgfpathlineto{\pgfpoint{\pgf@xb + 5pt - \@tempdimb}{\pgf@ya}}
\pgfpathlineto{\pgfpoint{\pgf@xa - 5pt + \@tempdima}{\pgf@ya}}
\pgfpathclose
}
}}

\boxshape{NEbox}{0pt}{5pt}
\boxshape{SEbox}{0pt}{-5pt}
\boxshape{NWbox}{5pt}{0pt}
\boxshape{SWbox}{-5pt}{0pt}
\boxshape{EBox}{-3pt}{3pt}
\boxshape{WBox}{3pt}{-3pt}
\makeatother

\tikzstyle{cloud}=[shape=cloud,draw,minimum width=1.5cm,minimum height=1.5cm]

\tikzstyle{map}=[draw,shape=NEbox,inner sep=2pt,minimum height=6mm,fill=white]
\tikzstyle{dashedmap}=[draw,dashed,shape=NEbox,inner sep=2pt,minimum height=6mm,fill=white]
\tikzstyle{mapdag}=[draw,shape=SEbox,inner sep=2pt,minimum height=6mm,fill=white]
\tikzstyle{mapadj}=[draw,shape=SEbox,inner sep=2pt,minimum height=6mm,fill=white]
\tikzstyle{maptrans}=[draw,shape=SWbox,inner sep=2pt,minimum height=6mm,fill=white]
\tikzstyle{mapconj}=[draw,shape=NWbox,inner sep=2pt,minimum height=6mm,fill=white]

\tikzstyle{medium map}=[draw,shape=NEbox,inner sep=2pt,minimum height=6mm,fill=white,minimum width=7mm]
\tikzstyle{medium map dag}=[draw,shape=SEbox,inner sep=2pt,minimum height=6mm,fill=white,minimum width=7mm]
\tikzstyle{medium map adj}=[draw,shape=SEbox,inner sep=2pt,minimum height=6mm,fill=white,minimum width=7mm]
\tikzstyle{medium map trans}=[draw,shape=SWbox,inner sep=2pt,minimum height=6mm,fill=white,minimum width=7mm]
\tikzstyle{medium map conj}=[draw,shape=NWbox,inner sep=2pt,minimum height=6mm,fill=white,minimum width=7mm]
\tikzstyle{semilarge map}=[draw,shape=NEbox,inner sep=2pt,minimum height=6mm,fill=white,minimum width=9.5mm]
\tikzstyle{semilarge map trans}=[draw,shape=SWbox,inner sep=2pt,minimum height=6mm,fill=white,minimum width=9.5mm]
\tikzstyle{semilarge map adj}=[draw,shape=SEbox,inner sep=2pt,minimum height=6mm,fill=white,minimum width=9.5mm]
\tikzstyle{semilarge map dag}=[draw,shape=SEbox,inner sep=2pt,minimum height=6mm,fill=white,minimum width=9.5mm]
\tikzstyle{semilarge map conj}=[draw,shape=NWbox,inner sep=2pt,minimum height=6mm,fill=white,minimum width=9.5mm]
\tikzstyle{large map}=[draw,shape=NEbox,inner sep=2pt,minimum height=6mm,fill=white,minimum width=12mm]
\tikzstyle{large map conj}=[draw,shape=NWbox,inner sep=2pt,minimum height=6mm,fill=white,minimum width=12mm]
\tikzstyle{very large map}=[draw,shape=NEbox,inner sep=2pt,minimum height=6mm,fill=white,minimum width=17mm]

\tikzstyle{medium dmap}=[draw,doubled,shape=NEbox,inner sep=2pt,minimum height=6mm,fill=white,minimum width=7mm]
\tikzstyle{medium dmap dag}=[draw,doubled,shape=SEbox,inner sep=2pt,minimum height=6mm,fill=white,minimum width=7mm]
\tikzstyle{medium dmap adj}=[draw,doubled,shape=SEbox,inner sep=2pt,minimum height=6mm,fill=white,minimum width=7mm]
\tikzstyle{medium dmap trans}=[draw,doubled,shape=SWbox,inner sep=2pt,minimum height=6mm,fill=white,minimum width=7mm]
\tikzstyle{medium dmap conj}=[draw,doubled,shape=NWbox,inner sep=2pt,minimum height=6mm,fill=white,minimum width=7mm]
\tikzstyle{semilarge dmap}=[draw,doubled,shape=NEbox,inner sep=2pt,minimum height=6mm,fill=white,minimum width=9.5mm]
\tikzstyle{semilarge dmap trans}=[draw,doubled,shape=SWbox,inner sep=2pt,minimum height=6mm,fill=white,minimum width=9.5mm]
\tikzstyle{semilarge dmap adj}=[draw,doubled,shape=SEbox,inner sep=2pt,minimum height=6mm,fill=white,minimum width=9.5mm]
\tikzstyle{semilarge dmap dag}=[draw,doubled,shape=SEbox,inner sep=2pt,minimum height=6mm,fill=white,minimum width=9.5mm]
\tikzstyle{semilarge dmap conj}=[draw,doubled,shape=NWbox,inner sep=2pt,minimum height=6mm,fill=white,minimum width=9.5mm]
\tikzstyle{large dmap}=[draw,doubled,shape=NEbox,inner sep=2pt,minimum height=6mm,fill=white,minimum width=12mm]
\tikzstyle{large dmap conj}=[draw,doubled,shape=NWbox,inner sep=2pt,minimum height=6mm,fill=white,minimum width=12mm]
\tikzstyle{large dmap trans}=[draw,doubled,shape=SWbox,inner sep=2pt,minimum height=6mm,fill=white,minimum width=12mm]
\tikzstyle{large dmap adj}=[draw,doubled,shape=SEbox,inner sep=2pt,minimum height=6mm,fill=white,minimum width=12mm]
\tikzstyle{large dmap dag}=[draw,doubled,shape=SEbox,inner sep=2pt,minimum height=6mm,fill=white,minimum width=12mm]
\tikzstyle{very large dmap}=[draw,doubled,shape=NEbox,inner sep=2pt,minimum height=6mm,fill=white,minimum width=19.5mm]

\tikzstyle{muxbox}=[draw,shape=rectangle,minimum height=3mm,minimum width=3mm,fill=white]
\tikzstyle{dmuxbox}=[muxbox,doubled]

\tikzstyle{box}=[draw,shape=rectangle,inner sep=2pt,minimum height=6mm,minimum width=6mm,fill=white]
\tikzstyle{dbox}=[draw,doubled,shape=rectangle,inner sep=2pt,minimum height=6mm,minimum width=6mm,fill=white]
\tikzstyle{dmap}=[draw,doubled,shape=NEbox,inner sep=2pt,minimum height=6mm,fill=white]
\tikzstyle{dmapdag}=[draw,doubled,shape=SEbox,inner sep=2pt,minimum height=6mm,fill=white]
\tikzstyle{dmapadj}=[draw,doubled,shape=SEbox,inner sep=2pt,minimum height=6mm,fill=white]
\tikzstyle{dmaptrans}=[draw,doubled,shape=SWbox,inner sep=2pt,minimum height=6mm,fill=white]
\tikzstyle{dmapconj}=[draw,doubled,shape=NWbox,inner sep=2pt,minimum height=6mm,fill=white]

\tikzstyle{ddmap}=[draw,doubled,dashed,shape=NEbox,inner sep=2pt,minimum height=6mm,fill=white]
\tikzstyle{ddmapdag}=[draw,doubled,dashed,shape=SEbox,inner sep=2pt,minimum height=6mm,fill=white]
\tikzstyle{ddmapadj}=[draw,doubled,dashed,shape=SEbox,inner sep=2pt,minimum height=6mm,fill=white]
\tikzstyle{ddmaptrans}=[draw,doubled,dashed,shape=SWbox,inner sep=2pt,minimum height=6mm,fill=white]
\tikzstyle{ddmapconj}=[draw,doubled,dashed,shape=NWbox,inner sep=2pt,minimum height=6mm,fill=white]

\boxshape{sNEbox}{0pt}{3pt}
\boxshape{sSEbox}{0pt}{-3pt}
\boxshape{sNWbox}{3pt}{0pt}
\boxshape{sSWbox}{-3pt}{0pt}
\tikzstyle{smap}=[draw,shape=sNEbox,fill=white]
\tikzstyle{smapdag}=[draw,shape=sSEbox,fill=white]
\tikzstyle{smapadj}=[draw,shape=sSEbox,fill=white]
\tikzstyle{smaptrans}=[draw,shape=sSWbox,fill=white]
\tikzstyle{smapconj}=[draw,shape=sNWbox,fill=white]

\tikzstyle{dsmap}=[draw,dashed,shape=sNEbox,fill=white]
\tikzstyle{dsmapdag}=[draw,dashed,shape=sSEbox,fill=white]
\tikzstyle{dsmaptrans}=[draw,dashed,shape=sSWbox,fill=white]
\tikzstyle{dsmapconj}=[draw,dashed,shape=sNWbox,fill=white]

\boxshape{mNEbox}{0pt}{10pt}
\boxshape{mSEbox}{0pt}{-10pt}
\boxshape{mNWbox}{10pt}{0pt}
\boxshape{mSWbox}{-10pt}{0pt}
\tikzstyle{mmap}=[draw,shape=mNEbox]
\tikzstyle{mmapdag}=[draw,shape=mSEbox]
\tikzstyle{mmaptrans}=[draw,shape=mSWbox]
\tikzstyle{mmapconj}=[draw,shape=mNWbox]

\tikzstyle{mmapgray}=[draw,fill=gray!40!white,shape=mNEbox]
\tikzstyle{smapgray}=[draw,fill=gray!40!white,shape=sNEbox]

\makeatletter

\pgfdeclareshape{cornerpoint}{
\inheritsavedanchors[from=rectangle] 
\inheritanchorborder[from=rectangle]
\inheritanchor[from=rectangle]{center}
\inheritanchor[from=rectangle]{north}
\inheritanchor[from=rectangle]{south}
\inheritanchor[from=rectangle]{west}
\inheritanchor[from=rectangle]{east}
\backgroundpath{
\southwest \pgf@xa=\pgf@x \pgf@ya=\pgf@y
\northeast \pgf@xb=\pgf@x \pgf@yb=\pgf@y

\pgfmathsetmacro{\pgf@shorten@left}{\pgfkeysvalueof{/tikz/shorten left}}
\pgfmathsetmacro{\pgf@shorten@right}{\pgfkeysvalueof{/tikz/shorten right}}

\pgfpathmoveto{\pgfpoint{0.5 * (\pgf@xa + \pgf@xb)}{\pgf@ya - 5pt}}
\pgfpathlineto{\pgfpoint{\pgf@xa - 8pt + \pgf@shorten@left}{\pgf@yb - 1.5 * \pgf@shorten@left}}
\pgfpathlineto{\pgfpoint{\pgf@xa - 8pt + \pgf@shorten@left}{\pgf@yb}}
\pgfpathlineto{\pgfpoint{\pgf@xb + 8pt - \pgf@shorten@right}{\pgf@yb}}
\pgfpathlineto{\pgfpoint{\pgf@xb + 8pt - \pgf@shorten@right}{\pgf@yb - 1.5 * \pgf@shorten@right}}
\pgfpathclose
}
}

\pgfdeclareshape{cornercopoint}{
\inheritsavedanchors[from=rectangle] 
\inheritanchorborder[from=rectangle]
\inheritanchor[from=rectangle]{center}
\inheritanchor[from=rectangle]{north}
\inheritanchor[from=rectangle]{south}
\inheritanchor[from=rectangle]{west}
\inheritanchor[from=rectangle]{east}
\backgroundpath{
\southwest \pgf@xa=\pgf@x \pgf@ya=\pgf@y
\northeast \pgf@xb=\pgf@x \pgf@yb=\pgf@y

\pgfmathsetmacro{\pgf@shorten@left}{\pgfkeysvalueof{/tikz/shorten left}}
\pgfmathsetmacro{\pgf@shorten@right}{\pgfkeysvalueof{/tikz/shorten right}}

\pgfpathmoveto{\pgfpoint{0.5 * (\pgf@xa + \pgf@xb)}{\pgf@yb + 5pt}}
\pgfpathlineto{\pgfpoint{\pgf@xa - 8pt + \pgf@shorten@left}{\pgf@ya + 1.5 * \pgf@shorten@left}}
\pgfpathlineto{\pgfpoint{\pgf@xa - 8pt + \pgf@shorten@left}{\pgf@ya}}
\pgfpathlineto{\pgfpoint{\pgf@xb + 8pt - \pgf@shorten@right}{\pgf@ya}}
\pgfpathlineto{\pgfpoint{\pgf@xb + 8pt - \pgf@shorten@right}{\pgf@ya + 1.5 * \pgf@shorten@right}}
\pgfpathclose
}
}

\makeatother

\pgfkeyssetvalue{/tikz/shorten left}{0pt}
\pgfkeyssetvalue{/tikz/shorten right}{0pt}

\tikzstyle{kpoint common}=[draw,fill=white,inner sep=1pt,minimum height=4mm]
\tikzstyle{kpoint sc}=[shape=cornerpoint,kpoint common]
\tikzstyle{kpoint adjoint sc}=[shape=cornercopoint,kpoint common]
\tikzstyle{kpoint}=[shape=cornerpoint,shorten left=5pt,kpoint common]
\tikzstyle{kpoint adjoint}=[shape=cornercopoint,shorten left=5pt,kpoint common]
\tikzstyle{kpoint conjugate}=[shape=cornerpoint,shorten right=5pt,kpoint common]
\tikzstyle{kpoint transpose}=[shape=cornercopoint,shorten right=5pt,kpoint common]
\tikzstyle{kpoint symm}=[shape=cornerpoint,shorten left=5pt,shorten right=5pt,kpoint common]

\tikzstyle{wide kpoint sc}=[shape=cornerpoint,kpoint common, minimum width=1 cm]
\tikzstyle{wide kpointdag sc}=[shape=cornercopoint,kpoint common, minimum width=1 cm]

\tikzstyle{black kpoint}=[shape=cornerpoint,shorten left=5pt,kpoint common,fill=black,font=\color{white}]

\tikzstyle{black kpoint sm}=[shape=cornerpoint,shorten left=5pt,kpoint common,fill=black,font=\color{white},scale=0.75]

\tikzstyle{black kpoint adjoint}=[shape=cornercopoint,shorten left=5pt,kpoint common,fill=black,font=\color{white}]
\tikzstyle{black kpointadj}=[shape=cornercopoint,shorten left=5pt,kpoint common,fill=black,font=\color{white}]

\tikzstyle{black kpointadj sm}=[shape=cornercopoint,shorten left=5pt,kpoint common,fill=black,font=\color{white},scale=0.75]

\tikzstyle{black dkpoint}=[shape=cornerpoint,shorten left=5pt,kpoint common,fill=black, doubled,font=\color{white}]
\tikzstyle{black dkpoint adjoint}=[shape=cornercopoint,shorten left=5pt,kpoint common,fill=black, doubled,font=\color{white}]
\tikzstyle{black dkpointadj}=[shape=cornercopoint,shorten left=5pt,kpoint common,fill=black, doubled,font=\color{white}]

\tikzstyle{black dkpoint sm}=[shape=cornerpoint,shorten left=5pt,kpoint common,fill=black, doubled,font=\color{white},scale=0.75]
\tikzstyle{black dkpointadj sm}=[shape=cornercopoint,shorten left=5pt,kpoint common,fill=black, doubled,font=\color{white},scale=0.75] 

\tikzstyle{kpointdag}=[kpoint adjoint]
\tikzstyle{kpointadj}=[kpoint adjoint]
\tikzstyle{kpointconj}=[kpoint conjugate]
\tikzstyle{kpointtrans}=[kpoint transpose]

\tikzstyle{big kpoint}=[kpoint, minimum width=1.2 cm, minimum height=8mm, inner sep=4pt, text depth=3mm]

\tikzstyle{wide kpoint}=[kpoint, minimum width=1 cm, inner sep=2pt]
\tikzstyle{wide kpointdag}=[kpointdag, minimum width=1 cm, inner sep=2pt]
\tikzstyle{wide kpointconj}=[kpointconj, minimum width=1 cm, inner sep=2pt]
\tikzstyle{wide kpointtrans}=[kpointtrans, minimum width=1 cm, inner sep=2pt]

\tikzstyle{wider kpoint}=[kpoint, minimum width=1.25 cm, inner sep=2pt]
\tikzstyle{wider kpointdag}=[kpointdag, minimum width=1.25 cm, inner sep=2pt]
\tikzstyle{wider kpointconj}=[kpointconj, minimum width=1.25 cm, inner sep=2pt]
\tikzstyle{wider kpointtrans}=[kpointtrans, minimum width=1.25 cm, inner sep=2pt]

\tikzstyle{gray kpoint}=[kpoint,fill=gray!50!white]
\tikzstyle{gray kpointdag}=[kpointdag,fill=gray!50!white]
\tikzstyle{gray kpointadj}=[kpointadj,fill=gray!50!white]
\tikzstyle{gray kpointconj}=[kpointconj,fill=gray!50!white]
\tikzstyle{gray kpointtrans}=[kpointtrans,fill=gray!50!white]

\tikzstyle{gray dkpoint}=[kpoint,fill=gray!50!white,doubled]
\tikzstyle{gray dkpointdag}=[kpointdag,fill=gray!50!white,doubled]
\tikzstyle{gray dkpointadj}=[kpointadj,fill=gray!50!white,doubled]
\tikzstyle{gray dkpointconj}=[kpointconj,fill=gray!50!white,doubled]
\tikzstyle{gray dkpointtrans}=[kpointtrans,fill=gray!50!white,doubled]

\tikzstyle{white label}=[draw,fill=white,rectangle,inner sep=0.7 mm]
\tikzstyle{gray label}=[draw,fill=gray!50!white,rectangle,inner sep=0.7 mm]
\tikzstyle{black label}=[draw,fill=black,rectangle,inner sep=0.7 mm]

\tikzstyle{dkpoint}=[kpoint,doubled]
\tikzstyle{wide dkpoint}=[wide kpoint,doubled]
\tikzstyle{dkpointdag}=[kpoint adjoint,doubled]
\tikzstyle{wide dkpointdag}=[wide kpointdag,doubled]
\tikzstyle{dkcopoint}=[kpoint adjoint,doubled]
\tikzstyle{dkpointadj}=[kpoint adjoint,doubled]
\tikzstyle{dkpointconj}=[kpoint conjugate,doubled]
\tikzstyle{dkpointtrans}=[kpoint transpose,doubled]

\tikzstyle{kscalar}=[kpoint common, shape=EBox, inner xsep=-1pt, inner ysep=3pt,font=\small]
\tikzstyle{kscalarconj}=[kpoint common, shape=WBox, inner xsep=-1pt, inner ysep=3pt,font=\small]

\tikzstyle{spekpoint}=[kpoint sc,minimum height=5mm,inner sep=3pt]
\tikzstyle{spekcopoint}=[kpoint adjoint sc,minimum height=5mm,inner sep=3pt]

\tikzstyle{dspekpoint}=[spekpoint,doubled]
\tikzstyle{dspekcopoint}=[spekcopoint,doubled]

 \tikzstyle{upground}=[circuit ee IEC,thick,ground,rotate=90,scale=2.5]
 \tikzstyle{downground}=[circuit ee IEC,thick,ground,rotate=-90,scale=2.5]
 \tikzstyle{bigground}=[regular polygon,regular polygon sides=3,draw=gray,scale=0.50,inner sep=-0.5pt,minimum width=10mm,fill=gray]

\tikzstyle{arrs}=[-latex,font=\small,auto]
\tikzstyle{arrow plain}=[arrs]
\tikzstyle{arrow dashed}=[dashed,arrs]
\tikzstyle{arrow bold}=[very thick,arrs]
\tikzstyle{arrow hide}=[draw=white!0,-]
\tikzstyle{arrow reverse}=[latex-]
\tikzstyle{cdnode}=[]

\let\olddagger\dagger
\renewcommand{\dagger}{\ensuremath{\olddagger}\xspace}

\usepackage{makeidx}
\makeindex

\theoremstyle{definition}
\newtheorem{theorem}{Theorem}[section]
\newtheorem*{theorem*}{Theorem}
\newtheorem{corollary}[theorem]{Corollary}

\newtheorem{proposition}[theorem]{Proposition}

\newtheorem{definition}[theorem]{Definition}

\newtheorem{example}[theorem]{Example}

\newtheorem{example*}[theorem]{Example*}
\newtheorem{examples*}[theorem]{Examples*}
\newtheorem{remark}[theorem]{Remark}
\newtheorem{remark*}[theorem]{Remark*}

\if\lecturenotes0

\newtheorem{exer}[theorem]{Exercise}
\newtheorem{exeropt}[theorem]{Exercise}
\newtheorem{exer*}[theorem]{Exercise*}

\fi

\if\lecturenotes1

\newtheorem{exer*}[theorem]{Exercise*}

\newtheoremstyle{exercise}{3pt}{3pt}{\color{red}}{}{\bf}{}{.5em}{}
\theoremstyle{exercise}

\fi

\newcommand{\TODO}[1]{\marginpar{\scriptsize\bB \textbf{TODO:} #1\e}}

\newcommand{\TODOa}[1]{\marginpar{\scriptsize\bM \textbf{TODO:} #1\e}}
\newcommand{\TODOb}[1]{\marginpar{\scriptsize\bB \textbf{TODO:} #1\e}}

\newcommand{\COMMa}[1]{\marginpar{\scriptsize\bM \textbf{COMM:} #1\e}}
\newcommand{\COMMb}[1]{\marginpar{\scriptsize\bB \textbf{COMM:} #1\e}}

\newcommand{\CHECK}[1]{\marginpar{\scriptsize\bR \textbf{CHECK:} #1\e}}

\usepackage{color}
\def\bR{\begin{color}{red}} 
\def\bB{\begin{color}{blue}}
\def\bM{\begin{color}{magenta}}
\def\bC{\begin{color}{cyan}}
\def\bW{\begin{color}{white}}
\def\bBl{\begin{color}{black}} 
\def\bG{\begin{color}{green}}
\def\bY{\begin{color}{yellow}}
\def\e{\end{color}\xspace}
\newcommand{\bit}{\begin{itemize}}
\newcommand{\eit}{\end{itemize}\par\noindent}
\newcommand{\ben}{\begin{enumerate}}
\newcommand{\een}{\end{enumerate}\par\noindent}
\newcommand{\beq}{\begin{equation}}
\newcommand{\eeq}{\end{equation}\par\noindent}
\newcommand{\beqa}{\begin{eqnarray*}}
\newcommand{\eeqa}{\end{eqnarray*}\par\noindent}
\newcommand{\beqn}{\begin{eqnarray}}
\newcommand{\eeqn}{\end{eqnarray}\par\noindent}

\if\lecturenotes1

\renewcommand{\TODO}[1]{}
\renewcommand{\TODOa}[1]{}
\renewcommand{\TODOb}[1]{}
\renewcommand{\COMMa}[1]{}
\renewcommand{\COMMb}[1]{}
\renewcommand{\CHECK}[1]{}

\def\bR{\begin{color}{black}} 
\def\bB{\begin{color}{black}}
\def\bM{\begin{color}{black}}
\def\bC{\begin{color}{black}}
\def\bW{\begin{color}{black}}
\def\bG{\begin{color}{black}}
\def\bY{\begin{color}{black}}

\fi

\usepackage{tikzfig}

\title{Categorical Quantum Mechanics I:\\  Causal Quantum Processes}                                                                                                                                                                     
\author[1]{Bob Coecke}
\author[2]{Aleks Kissinger}

\date{}

\affil[1]{Department of Computer Science, Oxford. {\tt coecke@cs.ox.ac.uk}}          
\affil[2]{iCIS, Radboud University, Nijmegen. {\tt aleks@cs.ru.nl}}

\begin{document}     

\maketitle

\begin{abstract}
We derive the category-theoretic backbone of quantum theory  from a process  ontology. More specifically, we treat quantum theory as a theory of systems, processes and their interactions. 

In this first part of a three-part overview, we first present a general theory of diagrams, and in particular, of string diagrams.  We discuss why diagrams are a very natural starting point for developing scientific theories.  Then we define process theories, and define a very general notion of quantum type. We show how our process  ontology  enables us to assert causality, that is, compatibility of quantum theory and relativity theory, and we prove the no-signalling theorem.   

Other notable contributions include new, elegant derivations of the no-broadcasting theorem, unitarity of evolution, and Stinespring dilation, all for any `quantum' type in a general class of process theories.        
\end{abstract}    

\section{Introduction}    

This chapter is the first  of a three-part  overview  on categorical quantum mechanics (CQM), an area of applied category-theory that over the past twelve years or so has become increasingly prominent within physics, mathematics and computer science, and even has spin-offs in other areas such as computational linguistics. Probably the most appealing feature of CQM is the use of diagrams, which are related to the usual Hilbert space model via symmetric monoidal categories and structures therein.  However, we have written this overview in such a way that no prior knowledge on category theory is required. In fact,  it can be seen as a first encounter with the relevant parts of category theory.      

We start with boxes and wires, which together make up diagrams.  The wires stand for systems, and the boxes stand for processes.  Symmetric monoidal categories then arise when endowing diagrams with operations of sequential and parallel composition. There are very good reasons to start with diagrams, rather than with traditional category-theoretic axioms, one being that the set-theoretic underpinning of category theory invokes an additional level of bureaucracy, namely dealing with such details as the bracketing of expressions, which has no counterpart in the reality that one aims to describe.  In other words, the traditional symbolic presentation of monoidal categories suffers from a substantial overhead as compared to its diagrammatic counterpart, to the extent that  monoidal category theory itself  becomes much simpler if one takes diagrams as  a starting point.  From this perspective, the role of the traditional presentation of  monoidal categories is reduced to providing a bridge to standard mathematical models, for example, the presentation of quantum theory using Hilbert space.

Next we define very general \em quantum types\em.  Despite this generality, we are able to derive important results that characterise the behaviour of quantum systems, most notably, the \em no-broadcasting theorem \em \cite{Nobroadcast}.  The diagrammatic formalism also makes it remarkably easy to assert  compliance with the theory of relativity, namely, in terms of a \em causality postulate \em \cite{chiri1}, which, in category-theoretic terms, simply boils down to the tensor unit being terminal \cite{Cnonsig}.  Again, important results follow straightforwardly, such as \em unitarity of evolution \em and \em Stinesping dilation\em.

\paragraph{Other overviews, surveys and tutorials on CQM.}  This three-part overview 
has a big brother, namely, a forthcoming textbook \cite{CKbook}  by the same authors. This three-part overview,  which amounts to some 70 pages, strives to be self-contained, but the book (which is about 12 times that size) will provide a more comprehensive introduction suitable for students and researchers in a wide variety of disciplines and levels of experience.

There is  another forthcoming book \cite{HVbook}, but rather than providing a complete presentation of quantum theory  from first principles, it puts  many of the core aspects of CQM studied in earlier papers in one place, and  also emphasises how the structures used in CQM  appear in other areas of mathematics.

A tutorial that provides a pedestrian introduction to the relevant category theory for CQM, and is complementary to this three-part one, is \cite{CatsII}.   
  
There have been shorter overviews on CQM before in the spirit of this one, e.g.~\cite{Kindergarten, ContPhys}, but the material simply wasn't ripe enough at that time for a fully-comprehensive presentation of quantum theory. 

\paragraph{Comparison to previous versions of CQM.}  The most prominent difference between this overview and the way CQM has been presented in the past is the fact that we take diagrams as our starting point, rather than standard category-theoretic structures and consistently introduce all new ingredients of the formalism in those terms.  While process ontology has played a motivational role throughout the entire CQM endeavour, see e.g.~\cite{JTF}, only now does every single ingredient emerge from this ontology.  This idea to start with diagrams as a primitive notion, even when a symbolic alternative is available, has also been advocated by Hardy \cite{HardyJTF}.

As compared to the paper by Abramsky and Coecke that initiated CQM \cite{AC1}, a  now well-established difference is  that we no longer rely on biproducts  (a non-diagrammatic concept) for the description of classical types.  They have been superseded  by special processes called \textit{spiders},  which allow one to express the characteristic processes associated with classical data (most importantly: `copy' and `delete') and reason about them diagrammatically \cite{CPav, CPaqpav}.

A more recent  innovation is the central role  now played by Chiribella, D'Ariano and Perinotti's causality postulate \cite{chiri1}.  In the past, issues of normalisation were mostly ignored in CQM.  Only recently it became apparent that the type of normalisation captured in the causality postulate is something structurally very fundamental  to processes which are physically realisable. Notably, all the results in Section \ref{sec:causality} of this chapter rely on it.  


\section{Process theories}

\begin{quote}
\em The art of progress is to preserve order amid change, and to preserve change amid order.\par \em \hfill    --- Alfred North Whitehead, Process and Reality, 1929.       
\end{quote}

\noindent
After his attempt to complete the set-theoretic foundations of mathematics in collaboration with Russell, Whitehead's venture into the natural sciences made him realise  that the traditional timeless ontology of substances, and not in the least their static set-theoretic underpinning, does not suit natural phenomena.  Instead, he claims, it is processes and their relationships which should underpin our understanding of those phenomena.

Can one turn this stance into a formal underpinning for natural sciences?  Category theory is a big step in that direction, but falls short of shaking off its set-theoretic  shackles.


\subsection{Processes as diagrams}\label{sec:procs-as-diagrams}

We shall use the term \em process \em to refer to anything that has zero or more inputs and zero or more outputs. For instance, the function:
\beq\label{eq:fstfunctf}
f:  {\mathbb R}\times {\mathbb R} \to {\mathbb R}:: (x, y) \mapsto  x^2 + y
\eeq
is a process which takes two real numbers as input and produces one real number as output. We represent such a process as a \textit{box} with some \textit{wires}  coming in the bottom to represent input systems and some wires coming out the top to represent output systems. For example, we could write the function  (\ref{eq:fstfunctf}) like this:  
\begin{equation}\label{eq:R-function-ex}
  \input{./figures/f-ex.tikz}
\end{equation}
The labels on wires are called \textit{system-types} or simply \textit{types}.

Similarly, a computer program is a process which takes some data (e.g. from memory) as input and produces some new data as output. For example, a program that sorts lists might look like this:
\ctikzfig{quicksort-ex} 
The following are also perfectly good processes:
\[
\input{./figures/physical-processes1.tikz}
\qquad\quad \ \, \ \input{./figures/physical-processes2.tikz}    
\qquad\qquad \input{./figures/physical-processes3.tikz} 
\]
Clearly the world around us is packed with processes!

We can \textit{wire together} simple processes to make more complicated processes,  which are described by \em diagrams\em:
\ctikzfig{compound-process}
We form diagrams by plugging the outputs of some processes into the inputs of others. This is allowed only if the types of the output and the input match.  For example, the following two processes: 
\ctikzfig{type-restriction1}
can be connected in some ways, but not in others, depending on the types of their wires:
\ctikzfig{type-restriction2}
This restriction on which wirings are allowed is an essential part of the language of diagrams,  in that it tells us when it makes sense to apply a process to a certain system and prevents occurrences like this:
\[
\quad\input{./figures/poosort.tikz} 
\]
which  probably wouldn't be very good for your computer!    

One thing to note is we haven't yet been too careful to say 
what a diagram actually \underline{is}. A complete description of a diagram consists of:  
\ben
\item what boxes it contains, and 
\item how those boxes are connected.  
\een
So the diagram refers to the `drawing' of boxes and wires without the interpretation of the diagram as a process. However, it makes no reference to where boxes are written on the page.  Hence, if two diagrams can be deformed into each other (without changing connections of course) then they are equal:  
\begin{equation}\label{eq:circuit-equiv}
  \input{./figures/circuit_diagram_equiv2_BOOK.tikz}
\end{equation}

When we think about interpreting the boxes in a diagram as processes, we are usually not interested in all possible processes, but rather in a certain class of related processes.  For example, practicioners of a particular scientific discipline will typically only study a particular class of processes: physical processes, chemical processes, biological processes, computational processes, etc. For that reason, we organise processes into \textit{process theories}. 

\begin{definition}\label{def:process-theory}
  A \textit{process theory} consists of:
  \begin{enumerate}
    \item[(i)] a collection $T$ of \textit{system-types} represented by wires,
    \item[(ii)] a  collection $P$ of \textit{processes} represented by boxes, where for each process in $P$ the input types and output types are taken from $T$, and
    \item[(iii)] a means of `wiring processes together'. That is, an operation that interprets a diagram of processes in $P$ as a process in $P$.
  \end{enumerate}
\end{definition}

\noindent In particular, (iii) guarantees that: 
\begin{center}
\em process theories are `closed under wiring processes together',  \em
\end{center}
since it is this operation that tells us what `wiring processes together' means. In some cases this operation consists of literally plugging things together with physical wires, like in the  theory of {\bf electrical devices}:  
\[
   \input{./figures/electricity-example2.tikz} \quad := \quad
  \raisebox{-2.7cm}{\includegraphics[width=2.4cm]{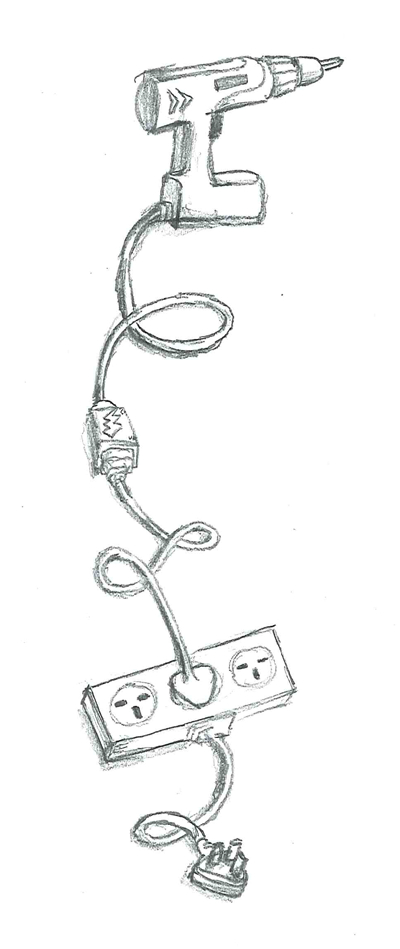}}
\]
In other cases this will require some more work, and sometimes there is more than one obvious choice available. We shall see below that in traditional mathematical practice one typically breaks down `wiring processes together' in two sub-operations: parallel composition and sequential composition of processes.

According to Definition~\ref{def:process-theory} a process theory tells us how to \textit{interpret} the boxes and wires in a diagram. But, crucially, by doing so it also tells us which diagrams consisting of those processes should be considered \textit{equal}.  

For example, suppose we define a simple process theory for \textbf{computer programs}, where the types are data-types (e.g. integers, booleans, lists, etc.) and the processes are computer programs. Then, consider a short program, which takes a list as input and sorts it. It might be defined this way (don't worry if you can't read the code, neither can half of the authors):
\[
\begin{tikzpicture}
	\begin{pgfonlayer}{nodelayer}
		\node [style=none] (0) at (0, -1.25) {};
		\node [style={medium box}] (1) at (0, 0) {\texttt{quicksort}};
		\node [style=none] (2) at (0, 1.25) {};
	\end{pgfonlayer}
	\begin{pgfonlayer}{edgelayer}
		\draw (0.center) to (1);
		\draw (1) to (2.center);
	\end{pgfonlayer}
\end{tikzpicture} \ \ \ := \ \ \ 
\begin{cases} 
  \,\verb!qs [] = []!\\
  \,\verb!qs (x :: xs) = !\\
  \,\verb!  qs [y | y <- xs; y < x] ++ [x] ++! \\
  \,\verb!  qs [y | y <- xs; y >= x]!
\end{cases}
\]
Wiring together programs means sending the output of one program to the input of another program. Taking two programs to be equal if they behave the same (disregarding some details like execution time, etc.), our process theory yields equations like this one:
\ctikzfig{quicksort_idempotent}
i.e.~sorting a list twice has the same effect as sorting it once. Unlike equation~\eqref{eq:circuit-equiv}, this is not just a matter of deforming one diagram into another, but actually represents a non-trivial equation between processes.

The reason we call a process theory a \textit{theory} is that it comes with lots of such equations, and these equations are precisely what allows us to draw conclusions about the processes we are studying.  Another thing one expects from a theory is that it makes predictions. In particular, one usually expects a theory to produce some numbers (e.g.~probabilities) that can verified by experiments. Toward that end, we first identify two special kinds of processes:
\bit
\item \em States \em are processes without any inputs. In  operational terms they are `preparation procedures'. We represent them as follows:
\[
\pointmap{\psi}
\]
\item
\em Effects \em are processes without any outputs. We have borrowed this terminology from quantum theory, where effects play a key role.  From now on we  represent them as follows:
\[
\copointmap{\,\pi\,} 
\]
\eit
When composing a state and an effect a third special kind of process arises which neither has inputs nor outputs, called a \em number\em.  Every process theory has at least one number, given by the `empty diagram':
\[ 
1 \ :=\ \emptydiag 
\]
We label this number `$1$' because combining the empty diagram with any other process yields the same process again.

It is natural to interpret the number arising from the composition of a state and an effect as the \textit{probability} that, given the system is in that state, the effect happens: 
\ctikzfig{state_test_paper}
We refer to this procedure as the \textit{generalised Born rule}.

In the introduction we announced that we would introduce category-theoretic definitions after we introduced the relevant diagrammatic notion.  However, the simple notion of a diagram introduced above, has a not so simple category-theoretic description.  We will now first introduce a more sophisticated notion of a diagram, which actually has a simpler category-theoretic counterpart, and this will serve as a stepping stone to defining the category-theoretic counterpart to the diagrams considered above.


\subsection{Circuit diagrams}

Given that  boxes represent  processes, we can define two basic \em composition operations \em on processes with the following interpretations:
\begin{align*}
  \textrm{$f \otimes g$}
  & := 
  \textrm{`process $f$ takes place \underline{while} process $g$ takes place'} \\  
  \textrm{$f \circ g$\,}
  & :=
  \textrm{`process $f$ takes place \underline{after} process $g$ takes place'}  
\end{align*}

The  \em parallel composition \em operation, written `$\otimes$', consists of placing a pair of diagrams side-by-side: 
\[
\left( \ \input{./figures/parallel-composition1.tikz} \right) \ \otimes \ \left( \ \input{./figures/parallel-composition2.tikz} \right) \ \ := \ \ \input{./figures/parallel-composition3.tikz}   
\] 
Any two diagrams can be composed in this manner, since placing diagrams side-by-side does not involve connecting anything.  This reflects the fact that both processes are happening independently of each other.   

This composition operation is associative: 
\beq\label{tensorassoc}
\left( \boxmap{f} \otimes  \boxmap{g} \right) \otimes  \boxmap{h}\ =\   \boxmap{f} \  \boxmap{g}  \ \boxmap{h}\ =\ \boxmap{f} \otimes \left( \boxmap{g} \otimes \boxmap{h} \right)
\eeq
and it has a unit, the empty diagram:
\beq\label{tensorunit}
\input{./figures/empty-diagram-unit.tikz}
\eeq

Parallel composition  is defined for system-types as well. That is, for types $A$ and $B$, we can form a new type $A \otimes B$, called the \textit{joint system-type}:
\ctikzfig{parallel-compose-wires}
There is also a special `empty' system-type, symbolically denoted $I$, which is used to represent `no inputs', `no ouputs' or both.

The \em sequential composition \em operation, written `$\circ$', consists of connecting the outputs of one diagram to the inputs of another diagram:
\[
\left(\ \input{./figures/sequential-composition1.tikz}\right)\ \circ \ \left(\ \input{./figures/sequential-composition2.tikz}\right)\ \ = \ \ \input{./figures/sequential-composition3.tikz}
\]
In other words, the process on the right  happens first,  followed by the process on the left,  taking the output of the first process as its input. Clearly not any pair of diagrams can be composed in this manner: the number and type of the inputs of the left process must match the number and type of the outputs of the right process.

The sequential composition operation is also associative: 
\beq\label{compassoc}
\left( \boxmap{h} \circ  \boxmap{g} \right) \circ  \boxmap{f}\   =     \ 
\input{./figures/threechain.tikz} \ 
=  \ \boxmap{h} \circ \left( \boxmap{g} \circ \boxmap{f} \right)
\eeq
and it also has a unit.  This time, it's a plain wire of appropriate type:
\beq\label{compunit}
\input{./figures/identity-diagram-unit.tikz}
\eeq

These two composition operations, due to their diagrammatic origin, also obey an \em interchange law\em.  Since we have:
\[
\left(\boxmap{g_1}\otimes \boxmap{g_2}\right)\circ\left(\boxmap{f_1}\otimes \boxmap{f_2}\right)
=
\left(\boxmap{g_1}\ \boxmap{g_2}\right)\circ\left(\boxmap{f_1}\ \boxmap{f_2}\right)
=\ 
\raisebox{0.6mm}{\input{./figures/twochain1.tikz}\  \input{./figures/twochain2.tikz}}
\]
\[
\left(\boxmap{g_1}\circ \boxmap{f_1}\right)\otimes\left(\boxmap{g_2}\circ \boxmap{f_2}\right)
=
\left(\,\raisebox{0.6mm}{\input{./figures/twochain1.tikz}}\right)  \otimes \left(\,\raisebox{0.6mm}{\input{./figures/twochain2.tikz}}\right)
\ =\ 
\raisebox{0.6mm}{\input{./figures/twochain1.tikz}\  \input{./figures/twochain2.tikz}}
\]
it follows that:
\beq\label{eq:bifunct}
\left(\boxmap{g_1}\otimes \boxmap{g_2}\right)\circ\left(\boxmap{f_1}\otimes \boxmap{f_2}\right)
=
\left(\boxmap{g_1}\circ \boxmap{f_1}\right)\otimes\left(\boxmap{g_2}\circ \boxmap{f_2}\right)
\eeq

Note that $\circ$ assumes there is some ordering on the input/output wires, and plugs them together `in order'. We can then express different orders by means of two wires crossing over each other:
\ctikzfig{specialbox2bisbis}
which as called a \textit{swap}.

\begin{definition} 
A diagram is a  \em circuit \em  if it  can be constructed  by composing boxes, including identities and swaps, by means of $\otimes$ and $\circ$.
\end{definition}

Every diagram that we have seen so far in this chapter is in fact a circuit.    Here is an example of the assembly of such a circuit:
\ctikzfig{circuitassembly} 
But not all diagrams are circuits. To understand  which ones are, we provide an equivalent characterisation that doesn't refer to the manner that one can build these diagrams, but in terms of a property that has to be satisfied.  

\begin{definition}
  A \textit{directed path} of wires is a list of wires $(w_1, w_2, \ldots, w_n)$ in a diagram such for all $i < n$, the wire $w_i$ is an input to some box for which 
  the wire $w_{i+1}$ is an output.   A \textit{directed cycle} is a directed path that starts and ends at the same box. 
\end{definition}

An example of a directed path is shown in bold here:
\ctikzfig{directed-path}
and an an example of a directed cycle is:
\ctikzfig{directed-cycle}
While diagrams only allow inputs to be connected to outputs, we indeed have done nothing (so far) to rule out the existence of directed cycles. This is precisely what the restricting to `circuit diagrams' does for us:  

\begin{theorem}\label{thm:circuit-acyclic}
The following are equivalent:
\bit
\item a diagram is a circuit, and, 
\item it contains no directed cycles.
\eit
\end{theorem}

\subsection{Category-theoretic counterpart}

We are now ready to provide a category-theoretic counterpart to diagrams.  A  \em symmetric monoidal category \em is essentially a process theory where all diagrams are circuit diagrams. Systems and processes are renamed, respectively as \em objects \em  and \em morphisms\em, and rather than diagrams being the main actor, the composition operations $\circ$ and $\otimes$ are taken as primitive. In order to guarantee that these operations behave as they did with diagrams, we have to require extra equations, namely equations (\ref{tensorassoc})--(\ref{eq:bifunct}) above.

\begin{definition}\label{def:strict-monoidal}
  A (strict) \textit{monoidal category} $\mathcal C$ consists of: 
  \begin{itemize}
    \item a collection $\textrm{ob}(\mathcal C)$ of \textit{objects},
    \item for every pair of objects $A, B$, a set $\mathcal C(A,B)$ of \textit{morphisms},
    \item for every object $A$, a special \em identity morphism \em $1_A \in \mathcal C(A,A)$,
    \item a \em sequential composition \em operation for morphisms:
    \[ (- \circ -) : \mathcal C(B,C) \times \mathcal C(A,B) \to \mathcal C(A,C) \]
    \item a \em parallel composition \em operation for objects: 
\[  
(- \otimes -) : \textrm{ob}(\mathcal C) \times \textrm{ob}(\mathcal C) \to \textrm{ob}(\mathcal C)
\]
    \item a \em unit object \em $I \in \textrm{ob}(\mathcal C)$, and
    \item a \em parallel composition \em operation for morphisms:
    \[ 
    (- \otimes -) : \mathcal C(A,B) \times \mathcal C(C,D) \to \mathcal C(A\otimes C,B \otimes D) 
    \]  
  \end{itemize}
  such that:
  \begin{itemize}
    \item $\otimes$ is associative and unital on objects:
    \[ (A \otimes B) \otimes C = A \otimes (B \otimes C) \qquad A \otimes I = A = I \otimes A \]
    \item $\otimes$ is associative and unital on morphisms:
    \[ \ \  (f \otimes g) \otimes h = f \otimes (g \otimes h) \ \   \qquad f \otimes 1_I = f = 1_I \otimes f \]
    \item $\circ$ is associative and unital on morphisms:
    \[ \ \ \ \   (h \circ g) \circ f = h \circ (g \circ f) \ \ \ \   \qquad 1_B \circ f = f = f \circ 1_A \]
    \item $\otimes$ and $\circ$ satisfy the \textit{interchange law}:  
    \[ 
    (g_1\otimes g_2)\circ(f_1\otimes f_2) = (g_1\circ f_1)\otimes(g_2\circ f_2)
    \]
  \end{itemize}
\end{definition}

Note we often write $f : A \to B$ as shorthand for $f \in \mathcal C(A, B)$, which indicates that we are thinking of morphisms as maps of some kind.

\begin{definition}
  A \textit{symmetric monoidal category} (SMC) is a monoidal category with a swap morphism:
  \[
  \sigma_{A,B} : A\otimes B \to B \otimes A
  \]
  defined for all objects $A, B$, satisfying:  
  \begin{itemize}
    \item $\sigma_{B,A} \circ \sigma_{A,B} = 1_{A\otimes B}$ 
    \item $(f \otimes g) \circ \sigma_{A,B} = \sigma_{B',A'} \circ (g \otimes f)$
    \item $\sigma_{A,I} = 1_A$
    \item $(1 \otimes \sigma_{A,C}) \circ (\sigma_{A,B} \otimes 1_C) = \sigma_{A,B\otimes C}$ 
  \end{itemize}
\end{definition} 

The diagrammatic counterparts to the first two of these equations are:
\ctikzfig{circuit_diagram_equiv_paper}
which are simply diagram deformations. The remaining two equations are again tautologies in terms of diagrams:
\ctikzfig{circuit_diagram_equiv_paper2}
  
It goes without saying that the category-theoretic definition is much more involved than its diagrammatic counterpart.  And in fact, it gets worse when we drop the `strict' from it, as explained in \cite{CatsII} \S  3.4.4. Simply stated, dropping strictness allows for some of the defining equations not to hold on-the-nose,  but are allowed some `wiggle room'. However, making this precise requires a number of \em natural isomorphisms\em, which express `how' one  `wiggles' from the LHS to the RHS, together with a bunch of \em coherence conditions \em which force all the natural isomorphisms to fit together well. The reason why it does make sense (and is in fact necessary) to consider this very involved definition is that  pretty much all set-theoretic structures which organise themselves into monoidal categories are in fact of the non-strict variety.  However, there is a standard procedure for turning every non-strict monoidal category into an equivalent strict one, so considering only strict monoidal categories (as we shall do throughout these papers) yields no real loss of generality.

\begin{example}
The category FHilb whose objects are finite-dimensional Hilbert spaces and whose morphisms are linear maps forms a (non-strict) SMC. Sequential composition is just the usual composition of linear maps, parallel composition is tensor product, and the unit object is the 1-dimensional Hilbert space $\mathbb C$.   The category Hilb of all Hilbert spaces and bounded linear maps also forms an SMC. However,  it does not form a compact closed category, which we'll define in the next section, so in categorical quantum mechanics, we typically focus on FHilb.
\end{example}

The precise connection between SMCs and circuit diagrams is the following:

\begin{theorem}\label{thm:circuit-smc}
  Any circuit diagram can be interpreted as a morphism in an SMC, and two morphisms are equal according to the axioms of an SMC if and only if their circuit diagrams are equal.
\end{theorem}

We won't say that much about the categories that correspond to arbitrary diagrams,  since they will be superseded by the notions introduced in the next section. The main point is that one has to adjoin an additional operation, called \em trace\em, for every triple $A, B, C$ of objects:
\[     
{\rm tr}_{B,C}^A : \mathcal C(A\otimes B, A\otimes C)  \to \mathcal C(B,C)
\]
This trace obeys a bunch of axioms that guarantee the following counterpart diagrammatic behaviour:
\[
{\rm tr}_{B,C}^A ::\ \ \raisebox{-0.6mm}{\input{./figures/partialtracepre_paper.tikz}}\ \ \mapsto \ \ \input{./figures/partialtrace_paper.tikz}  
\]
While such a trace  does play a role in quantum theory,  we will soon see that it  arises as a derived concept in a simpler type of category.

\subsection{Reference and further reading}  

The use of diagrams started with Penrose's diagrammatic calculus for abstract tensor systems \cite{Penrose}. The proof that abstract tensor systems characterise the free traced symmetric monoidal category was given in~\cite{KissingerATS}.

Monoidal categories are due to Benabou \cite{Benabou}, with their modern formulation---and the fact that any monoidal category is equivalent to a strict monoidal category---being worked out by Mac Lane \cite{MacLaneCoherence}.  The connection between circuit diagrams and symmetric monoidal categories was established  by Joyal and Street in \cite{JS}, where they are referred to as `progressive diagrams'.
  

\section{String diagrams}  

\begin{quote}
\em When two systems, of which we know the states by their respective representatives, enter into temporary physical interaction due to known forces between them, and when after a time of mutual influence the systems separate again, then they can no longer be described in the same way as before, viz.~by endowing each of them with a representative of its own. I would not call that \textit{one} but rather \textit{the} characteristic trait of quantum mechanics, the one that enforces its entire departure from classical lines of thought.\par \em \hfill    --- Erwin Schr\"odinger,  1935.
\end{quote}

\noindent
By 1935, Schr\"odinger had already realised that the biggest gulf between quantum theory and our classical ways of thinking was really that, when it comes to quantum systems, the whole is more than the sum of its parts. In the classical world, for instance, it is possible to totally describe the state of two systems---say...two objects sitting on a table---by first totally describing the state of the first object then totally describing the state of the second object. This is a fundamental property one expects of a classical, or \textit{separable} (or, in category-theory lingo: \em Cartesian\em) universe. However, as Schr\"odinger points out, there exist states predicted by quantum theory (and observed in the lab!) which do not obey this `obvious' law about the physical world.  Schr\"odinger called this new, totally non-classical phenomenon \em verschr\"ankung\em, which later became translated to the dominant scientific language as \em entanglement\em. 

In contrast to this great insight, it took physicists many years to actually exploit this  property and reveal the features of the quantum world that are direct consequences of it.  Most strikingly, it took physicists some 60 years to discover what is probably the most direct consequence: \em quantum teleportation\em. 

\subsection{Non-separability}

\begin{definition}\label{def:stringdiagram} 
\em String diagrams \em  can be defined equivalently as: 
\bit
\item diagrams consisting of boxes and wires, where we additionally allow inputs to be connected to inputs and outputs  to be connected to outputs, for example:   
\[
\input{./figures/compound-process-capscups.tikz} 
\]
\eit
\bit
\item circuit diagrams that contain a special state and a special effect:
\[
\raisebox{0.2mm}{\input{./figures/CupStateType.tikz}}\! \qquad\quad {\rm and } \qquad\quad \raisebox{-0.2mm}{\input{./figures/CapEffectType.tikz}}
\]
for each type for which we have:
\begin{equation}\label{eq:cup-equations}
  \raisebox{2mm}{\input{./figures/CupCapNew.tikz}}\qquad\quad\qquad
\begin{array}{c}
\input{./figures/cap-box-swap.tikz} \, \ = \, \ \input{./figures/CapEffectType.tikz}\vspace{4mm}\\
\input{./figures/cup-box-swap.tikz}  \, \ = \, \ \input{./figures/CupStateType.tikz}
\end{array}
\end{equation}
\eit  
\end{definition}  

We can relate these two equivalent definitions by writing the cup state and cap effect from the second definition as cup- and cap-shaped pieces of wire:
\[
\raisebox{-1mm}{\begin{tikzpicture}
	\begin{pgfonlayer}{nodelayer}
		\node [style=none] (0) at (-0.75, 0.5) {};
		\node [style=none] (1) at (0.75, 0.5) {};
		\node [style=none] (2) at (0.75, 0.5) {};
		\node [style=none] (3) at (-0.75, 0.5) {};
	\end{pgfonlayer}
	\begin{pgfonlayer}{edgelayer}
		\draw [in=-90, out=-90, looseness=1.75] (0.center) to (1.center);
	\end{pgfonlayer}
\end{tikzpicture}} \ \ := \ \ \input{./figures/CupStateType.tikz} \qquad\qquad\raisebox{1mm} {\begin{tikzpicture}
	\begin{pgfonlayer}{nodelayer}
		\node [style=none] (0) at (-0.75, -0.5) {};
		\node [style=none] (1) at (0.75, -0.5) {};
		\node [style=none] (2) at (-0.75, -0.5) {}; 
		\node [style=none] (4) at (0.75, -0.5) {};
	\end{pgfonlayer}
	\begin{pgfonlayer}{edgelayer}
		\draw [in=90, out=90, looseness=1.75] (0.center) to (1.center);
	\end{pgfonlayer}
\end{tikzpicture}}\ \  :=\ \  \input{./figures/CapEffectType.tikz}  
\]
so that the equations in~\eqref{eq:cup-equations} become:
\begin{equation}\label{eq:wire-yank-all}
  \input{./figures/wire_yank_all.tikz}
\end{equation}
Hence they are reduced to simple wire deformations:  
\begin{center}
\includegraphics[height=3.5cm]{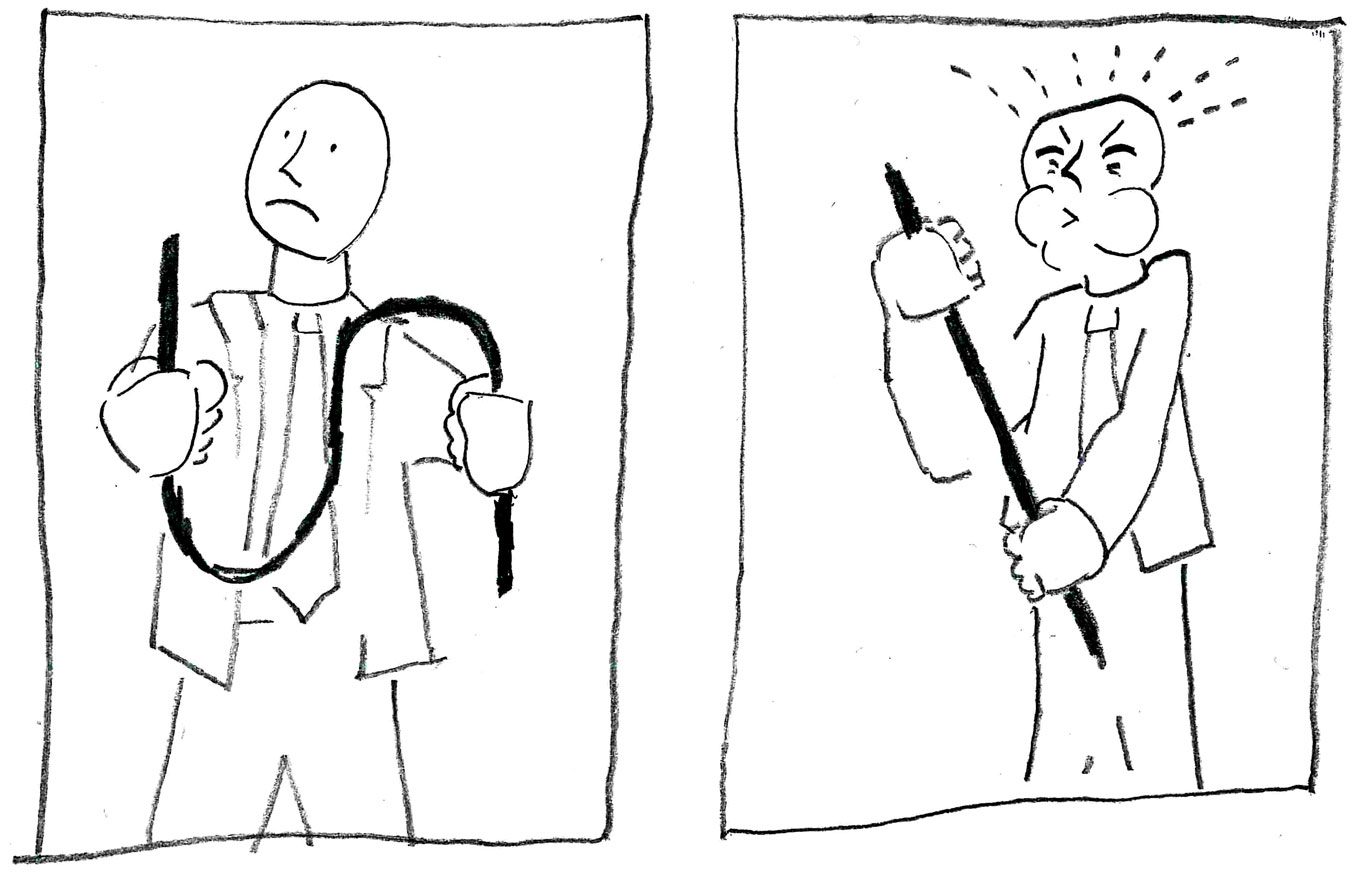}\qquad\  \includegraphics[height=3.5cm]{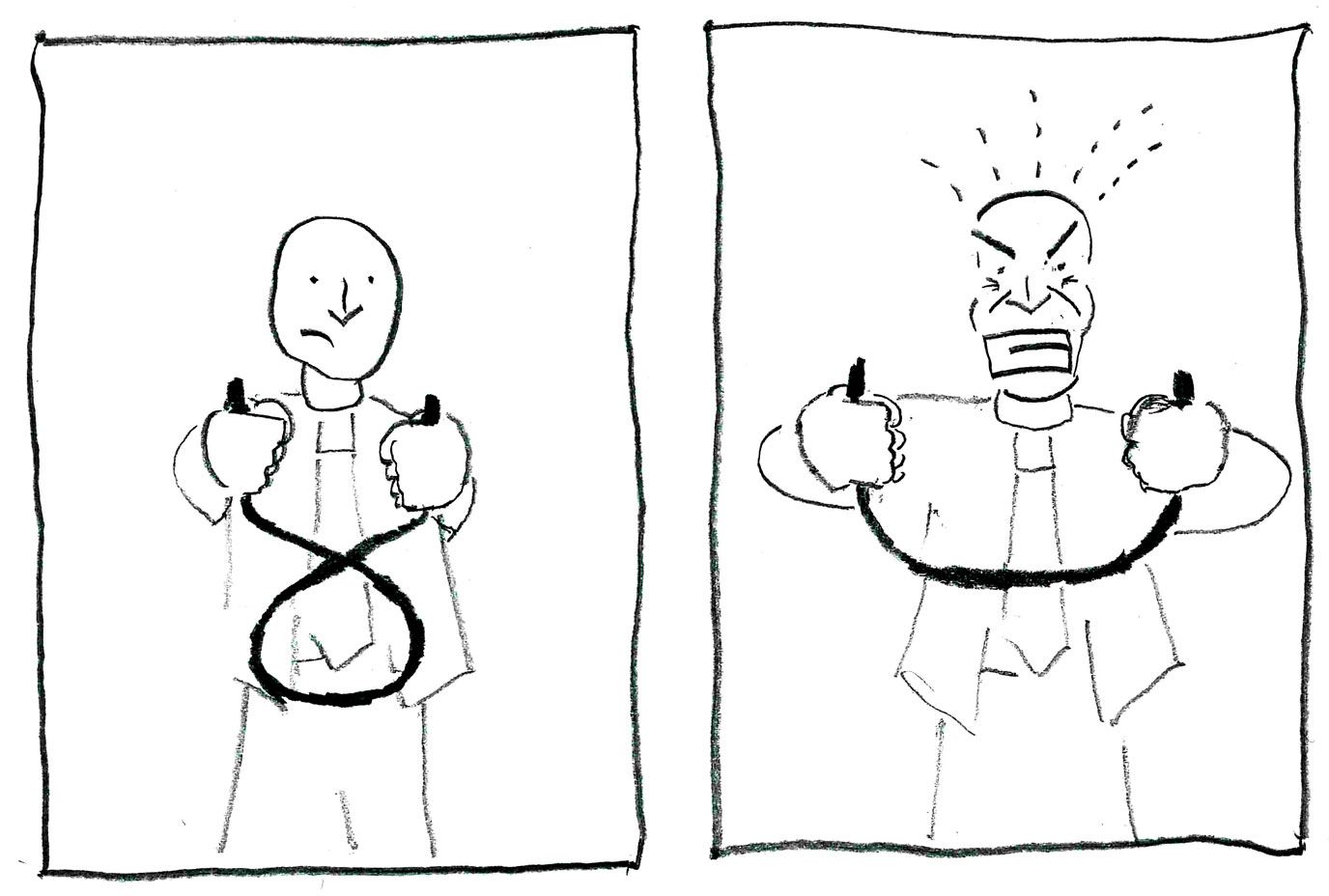}
\end{center}

String diagrams capture what Schr\"odinger described as the characteristic trait of quantum mechanics: the existence of non-separable states. Visualising cups and caps as wires captures this intuition that the two systems involved cannot be separated from each other (since, of course, we are not allowed to cut wires). We can state this more formally as follows:

\begin{proposition}\label{prop:cupdiscrubbish}
If a theory is described by string diagrams, and all two-system states $\psi$ are $\otimes$-separable,  i.e.~there exist states $\psi_1$ and $\psi_2$ such that:    
\ctikzfig{State2split}
then every process $f$ is $\circ$-separable, i.e.~it can be written as:   
\[
\boxmap{f}\ \  = \ \ 
\begin{array}{c}
\raisebox{1.9mm}{\pointmap{\,\psi}}\vspace{0.5mm}\\
\raisebox{-1.9mm}{\copointmap{\,\pi\,}}
\end{array}
\]
for some effect $\pi$ and some state $\psi$.
\end{proposition}
\begin{proof}
By assumption, the cup-state  is $\otimes$-separable:      
\[
\begin{tikzpicture}
	\begin{pgfonlayer}{nodelayer}
		\node [style=none] (0) at (-1, 0.5) {};
		\node [style=none] (1) at (1, 0.5) {}; 
	\end{pgfonlayer}
	\begin{pgfonlayer}{edgelayer}
		\draw [in=-90, out=-90, looseness=1.75] (0.center) to (1.center); 
	\end{pgfonlayer}
\end{tikzpicture}\ \ =\ \point{\psi_1} \point{\psi_2}  
\]
So, for any process $f$ we have:
\ctikzfig{Separationproof_paper}
\end{proof}

What is this telling us? First, we should note that all processes being $\circ$-separable is a \textit{bad thing}. It is the same as saying all processes in the theory simply throw away the input and produce a constant output (possibly multiplied by some number depending on the input). In other words, nothing ever happens!  As this is an absurd condition for  any `reasonable' process theory, all bipartite (i.e.~two-system) states cannot be $\otimes$-separable. 

In such reasonable theories, string diagrams guarantee that the collection of two-system states is just as rich as the processes involving the same types:    

\begin{proposition}\label{prop:closurestring}
Any theory described by string diagrams has \em process-state duality\em. That is, for any two types $A$ and $B$ we have a 1-to-1 correspondence between processes and states of the following form:
\[
\input{./figures/notrans.tikz}\ \qquad\longleftrightarrow\qquad\  \input{./figures/paper1.tikz}
\]
which is realised by:
\[
\boxmap{f}\ \ \mapsto\ \  \input{./figures/CupActionpaper.tikz} \qquad\qquad\ \   \input{./figures/State2psi.tikz}\ \ \mapsto\ \ \input{./figures/CapActionpaper.tikz}
\]
\end{proposition}

\subsection{Traces and transposes}    

For a  process $f$ with one of its inputs having the same type as one of its outputs, the \em partial trace \em with respect to that input/output pair is: 
\[
{\rm tr}_{B,C}^A \left(\input{./figures/partialtracepre_paper.tikz}\right)\ \, := \  \input{./figures/partialtrace_paper.tikz}
\]
The total trace, or simply the \textit{trace}, is the special case where $B = C = I$. A typical property of the trace is \textit{cyclicity},  which now comes just from  the observation that---since only connectivity matters---the following two diagrams are the same: 
\ctikzfig{trace_f_g_paper}

The \em transpose \em of a  process $f$ is the  process:  
\[
\input{./figures/yestrans.tikz} 
\]
or equivalently, by \eqref{eq:wire-yank-all}:
\[
\input{./figures/yestrans-flip.tikz} 
\]
Then, clearly if we transpose twice, we get back where we started: 
\[  
\input{./figures/transtransmapyes.tikz}\ \  =\ \ \boxmap{f} 
\]
In other words, transposition is an \textit{involution}.  And now comes the really cool bit. First, we deform our boxes a bit:
\[ 
\boxmap{f} \ \ \ \leadsto\ \ \ \map{f} 
\]
Now, if we express the transpose of $f$ as a box labelled `$f$', but rotated 180${}^\circ$:    
\[
\begin{tikzpicture}
	\begin{pgfonlayer}{nodelayer}
		\node [style=maptrans] (0) at (0, 0) {$f$};
		\node [style=none] (1) at (0, -1.25) {};
		\node [style=none] (2) at (0, 1.25) {};
	\end{pgfonlayer}
	\begin{pgfonlayer}{edgelayer}
		\draw [style=swap] (0) to (1.center); 
		\draw [style=swap] (2.center) to (0);
	\end{pgfonlayer}
\end{tikzpicture}\ \  :=\ \ \input{./figures/transmapyes.tikz}
\]
the definition of a transpose 
gets built-in to this `180${}^\circ$ rotation'-notation: 
\begin{center}
\includegraphics[height=4.2cm]{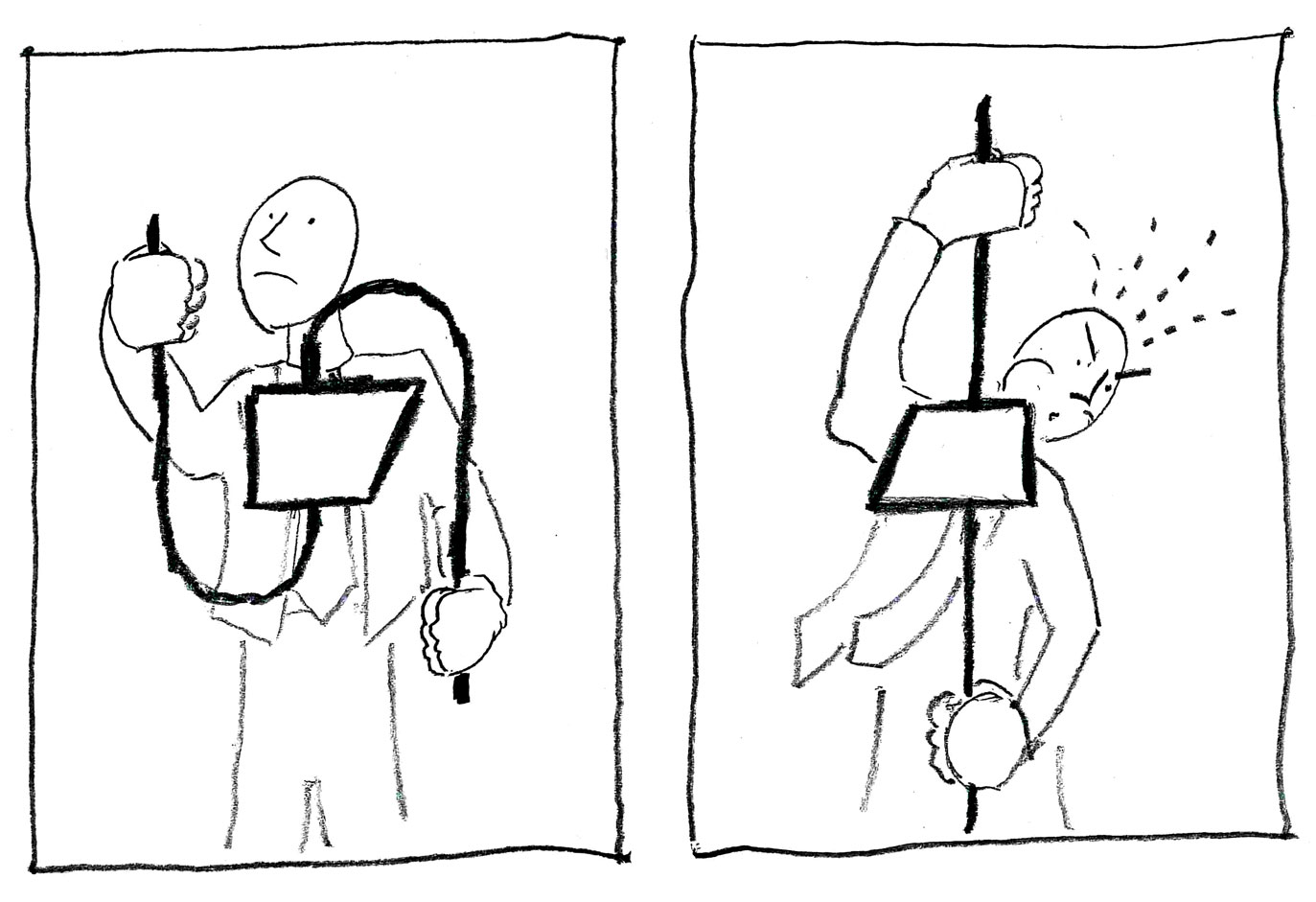}
\end{center}
From this, it also follows that:
\[
\input{./figures/slide1.tikz}\qquad\qquad\input{./figures/slide2.tikz}
\]
So we can slide boxes across cups and caps just like beads on a wire:
\[
\input{./figures/slide1steps.tikz}\qquad\quad \input{./figures/slide1steps2.tikz}
\]

Now, let's play a game. Here is the challenge: Aleks and Bob are far apart, Aleks possesses a system in a state $\psi$, and Bob needs this state.  Suppose they also share another state, namely the cup-state. So, we have this situation:  
\ctikzfig{telesum0NEWpre}
Starting from this arrangement, is there something Aleks and Bob can do in the `?'-marked regions indicated below,
that results in Bob obtaining $\psi$? 
\ctikzfig{telesum0NEW}
Here is a simple solution:
\ctikzfig{telesum00NEW}
If you the reader came up with this solution yourself, then you on your own did in a few seconds what all the physicists in the world failed to do between 1930 and 1990: discover quantum teleportation!  


\subsection{Adjoints and unitarity}\label{sec:adjoints-and-unitarity} 

There is a slight caveat in the above discussion of teleportation.  While the cap is a process that Aleks can  attempt to `do', he'll need a bit of luck to succeed. This is because the cap process is not \textit{causal} (a concept we'll soon meet), so it can only be realised with some probability strictly less than 1.  As a consequence, Aleks might not get the effect he wants (the cap), but some other effect  $\phi_1, \phi_2, \ldots$ which, by process-state duality, we can represent as a cap with  some `unwanted' box:
\[ 
\input{./figures/ps-U-paper.tikz}\  \ :=\ \ \input{./figures/ps-phi-paper.tikz} 
\]
which we'll called the \textit{error}. This box obstructs Bob's direct access to the state he so much desires:
\beq\label{eq:tele-with-error}
\input{./figures/telesum1-2NEWpapera.tikz} \ \ =\ \ \input{./figures/telesum1-2NEWpaperb.tikz}
\eeq
So, all is lost?  Of course not! Bob  just needs to figure out how to `undo'  this error by means of some process.

Before going there, we need to enrich our language of string diagrams a bit.  In addition to rotating boxes 180${}^\circ$, we will now also reflect them vertically:
\[
\input{./figures/smapAB.tikz} \ \ \mapsto \ \ \input{./figures/smapdagAB.tikz}  
\]
and refer to this reflected box as the \em adjoint\em.  Like transposition, this is  an involution, and it can be applied to entire diagrams in the obvious way:
\[
\input{./figures/arbitrary_klein.tikz}
\ \  \mapsto \ \ 
\input{./figures/arbitrary_klein_flip.tikz}
\]

Rotating by 180${}^\circ$ degrees then reflecting vertically is the same as just reflecting horizontally:
\[ 
\map{f} \ \mapsto\ \maptrans{f}\ \mapsto\ \mapconj{f} 
\]
This horizontal reflection, obtained by composing the transpose and the adjoint, is called the \em conjugate\em.  So all together, boxes now come in quartets:
\beq\label{eq:smapALLFOURnew}
\input{./figures/smapALLFOURnew.tikz}
\eeq
Having adjoints around enables us to make the following definition:

\begin{definition}
A process $U$ an \em isometry \em if we have:     
\beq\label{eq:isometry}
\input{./figures/isometry-types.tikz} 
\eeq
and it is \em unitary \em if its adjoint is also an isometry.    
\end{definition}

Each of these names may sound familiar to  some readers, which is of course no coincidence: 

\begin{example}
  The conjugate, transpose, and adjoint of a linear map $f$ is a new linear map obtained by taking the conjugate, transpose, or adjoint (a.k.a.~conjugate-tranpose) of the matrix of $f$, respectively. Hence isometries and unitaries of linear maps are just the standard notions.
\end{example}

We now return to \eqref{eq:tele-with-error}, where Bob seeks to undo the error introduced by Aleks' process. In the case where $U$ is an isometry (or better yet, a unitary), Bob simply needs to apply its adjoint to undo it:
\ctikzfig{telesum1NEW} 

\subsection{Adjoints and connectedness} 

We haven't said much about adjoints, just that they have to preserve diagrams.  In fact, the ontology of the adjoint and corresponding postulates are still the subject of ongoing research.  Below we will make use of one particular additional condition  that one may want  adjoints to satisfy.  This extra condition comes from taking the notion that an adjoint really is a reflection seriously.

Suppose we have a $\circ$-non-separable process.  One could then imagine  that it has some internal structure, say a collection of tubes or machines connecting some inputs to outputs: 
\begin{center}
  \includegraphics[width=2.5cm]{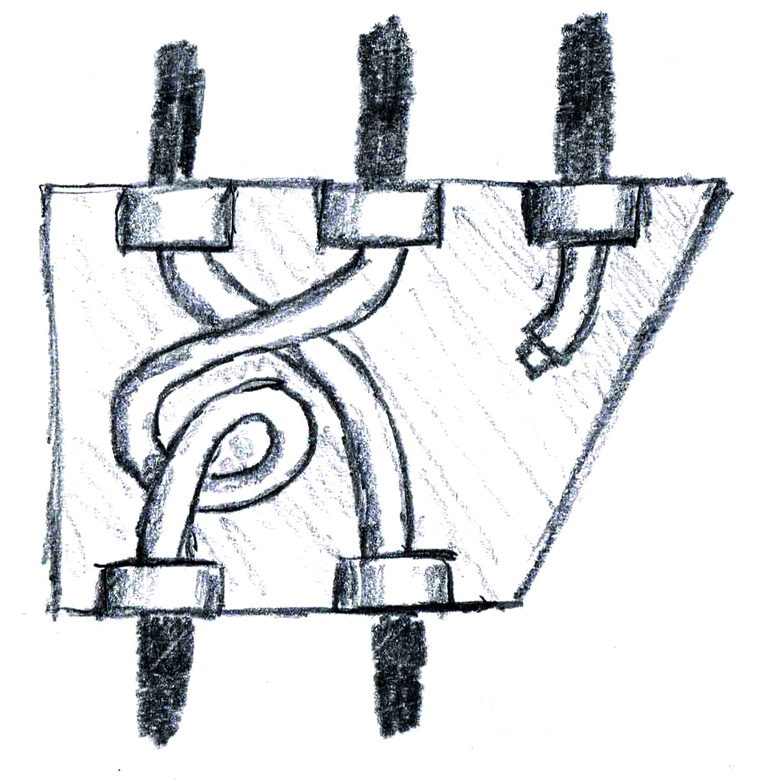}
\end{center}
If we now compose this process with its adjoint, i.e.~its vertical reflection, then these internal  connections match up:
\begin{center}
  \includegraphics[width=2.5cm]{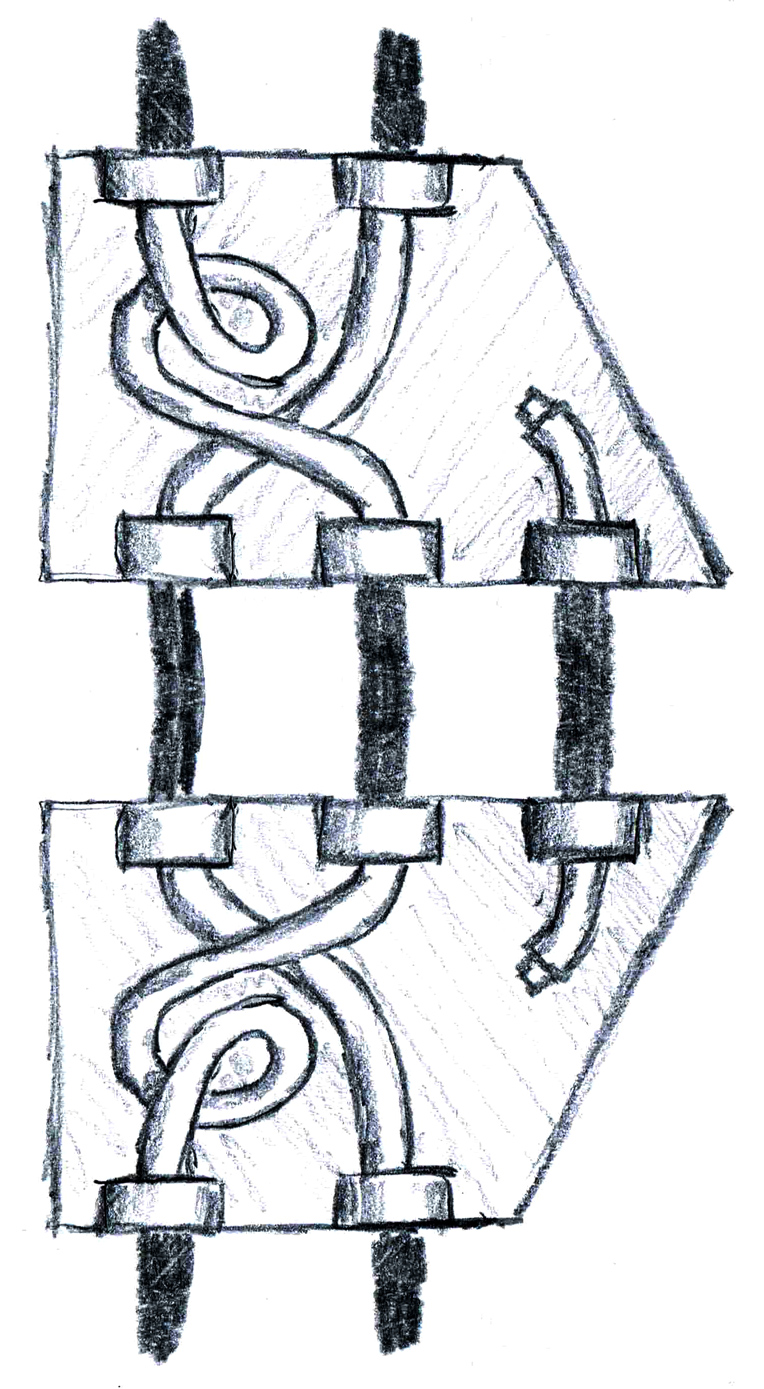}
\end{center}
so one expects the resulting process also to be $\circ$-non-separable, that is:
\beq\label{EQ:daggeraxiom}
\left( \exists \psi, \phi: \ \ \map{f} \ =\ \kpointketbra{\psi}{\phi} \right)
  \ \ \Longleftrightarrow\ \
 \left( \exists \psi', \phi':  \ \   \begin{tikzpicture}
	\begin{pgfonlayer}{nodelayer}
		\node [style=mapdag] (0) at (0, 1) {$f$};
		\node [style=map] (1) at (0, -1) {$f$};
		\node [style=none] (2) at (0, -2.25) {};
		\node [style=none] (3) at (0, 2.25) {};  
	\end{pgfonlayer}
	\begin{pgfonlayer}{edgelayer}
		\draw (1) to (2.center);
		\draw (3.center) to (0); 
		\draw (0) to (1);
	\end{pgfonlayer}
\end{tikzpicture}\ \ =\ \kpointketbra{\psi'}{\phi'}\right)  
\eeq

Indeed this assumption holds for our main example:
\begin{example}
  For linear maps, $\circ$-separable means rank-1. It is a well-known fact from linear algebra that the rank of:
  \ctikzfig{fdagf}
  is the same as the rank of $f$. Hence \eqref{EQ:daggeraxiom} is satisfied.
\end{example}

While this assumption makes sense visually, it can fail in surprising places:
\begin{example}
  Consider a process theory whose processes are relations $R \subseteq A \times B$ where $\circ$ is the usual composition of relations, $\otimes$ is the Cartesian product, and $I$ is the 1-element set $\{ * \}$. The adjoint $R^\dagger$ of a relation $R$ is its converse. That is, $(b,a) \in R^\dagger$ if and only if $(a,b) \in R$. Then the following relation from the two element set $\{0,1\}$ to itself fails to satisfy \eqref{EQ:daggeraxiom}:
  \[ 
  \{ (0,0), (0,1), (1,1) \} \subseteq \{0,1\} \times \{0,1\} 
  \]
\end{example}

\subsection{Category-theoretic counterpart}

In Definition~\ref{def:stringdiagram} we gave two equivalent definitions 
of string diagram, the latter of which consisting of circuit diagrams with special processes called cups and caps. The category-theoretic counterpart for string diagrams proceeds in pretty much the same way:

\begin{definition}\label{def:compact-closed}
  A \textit{compact closed category} (with symmetric self-duality) is a symmetric monoidal category $\mathcal C$ such that for every object $A \in \textrm{ob}(\mathcal C)$, there exists morphisms $\epsilon_A \in \mathcal C(A \otimes A, I)$ and $\eta_A \in \mathcal C(I, A \otimes A)$ such that:  
\[ 
\begin{array}{ccc}
& &  \epsilon_A \circ \sigma_{A,A} = \epsilon_A\\
(\epsilon_A \otimes 1_A) \circ (1_A \otimes \eta_A) = 1_A & \qquad\quad & \\
& & \sigma_{A,A} \circ \eta_A = \eta_A
\end{array}
  \]
\end{definition}

Compact closed categories satisfy an analogue to Theorem~\ref{thm:circuit-smc}, but for string diagrams:

\begin{theorem}\label{thm:string-cmpcc}
  Any string diagram can be interpreted as a morphism in a compact closed category, and two morphisms are equal according to the axioms of a compact closed category if and only if their string diagrams are equal.
\end{theorem}

\begin{example}\label{ex:hilb-compact}
The category FHilb is compact closed. For a Hilbert space $A$, fix a basis $\{ e_i \}_i$ in $A$ then let:
\[ 
\eta_A : \mathbb C \to A \otimes A
   \qquad\qquad\textrm{and}\qquad\qquad 
   \epsilon_A : A \otimes A \to \mathbb C
\]
be linear maps defined as follows:
\[
   \eta_{A}(1) = \smallsum\ e_i \otimes e_i
   \qquad\qquad
   \epsilon_{A}(e_i \otimes e_j) = \begin{cases}
    1 & \textrm{ if } i = j \\
    0 & \textrm{ otherwise}
   \end{cases}   
   \]
Using Dirac's  `bra-ket' notation,  which is popular in the quantum computing literature, these maps can be written as:
\[ 
\eta_A = \smallsum\ \ket{i} \otimes \ket{i}
   \qquad\qquad
   \epsilon_A = \smallsum\ \bra{i} \otimes \bra{i} 
\]
\end{example}

In category theory, closure refers to the fact that one can represent sets of morphisms $\mathcal C(A,B)$ as the states of another object, which we could denote as $A\Rightarrow B$.  Compact closed categories have the very convenient feature that we can take: 
\[
A\Rightarrow B := A\otimes B
\]
We saw this feature in Proposition \ref{prop:closurestring},  under the name `process-state duality'.

Just as we simplified before by restricting to \textit{strict} monoidal categories, here we simplify from a more general notion of compact closed categories to symmetrically self-dual compact closed categories. The former only require that $\epsilon_A \in \mathcal C(A^* \otimes A, I)$ and $\eta_A \in \mathcal C(I, A \otimes A^*)$ exists for some object $A^*$ (called the \textit{dual} of $A$). Then, `self-duality' means we can choose $A^*$ to just be $A$ again, and `symmetric' means this choice gets along well with symmetries.  This of course makes sense from a diagrammatic point of view (cf.~the right-most equations in~\eqref{eq:wire-yank-all}),  and as pointed out in~\cite{SelingerSelfDual}, is necessary for interpreting string diagrams as morphisms without ambiguity. 

\begin{remark}
  The definition of $\eta_A$ and $\epsilon_A$ from Example~\ref{ex:hilb-compact} depend on the choice of basis. We can avoid this by dropping self-duality, in which case we take $A^*$ to be the dual space and let:
  \[ \epsilon_A(\xi \otimes a) = \xi(a) \]
  This uniquely fixes $\epsilon_A$ (and therefore $\eta_A$) without making reference to a basis.
\end{remark}

So to summarise, we have the following correspondences between  diagrams and categories: 
\begin{center}
\begin{tabular}{|c|c|c|}     
\hline
\em circuit diagram \em & \em 
diagram \em & \em string diagram \em \\
\input{./figures/leiden3x.tikz} & \input{./figures/leiden1xx.tikz} & \input{./figures/compound-process-capscups-book.tikz} \\
i.e.~admits causal structure & i.e.~outputs to inputs & i.e.~all to all \\
\hline
$\Rightarrow$ plain SMC  &  $\Rightarrow$ traced SMC & $\Rightarrow$ compact closed C \\    
\hline
\end{tabular}
\end{center}
Notably, the more sophisticated  (i.e.~restrictive) kinds of diagrams correspond to the simplest kinds of SMCs, and vice versa.

Unsurprisingly, the category-theoretic definition of  adjoints is again more involved than its diagrammatic counterpart, simply for the reason that we now have to say carefully what `preserving diagrams under reflection' means in terms of the language of SMCs.

\begin{definition}
For a strict SMC $\mathcal C$, a \textit{dagger functor} assigns to each morphism $f : A \to B$ a new morphism $f^\dagger : B \to A$ such that:
\[ 
  (f^\dagger)^\dagger = f \qquad
  (g \circ f)^\dagger = f^\dagger \circ g^\dagger \qquad
  (f \otimes g)^\dagger = f^\dagger \otimes g^\dagger \qquad
  \sigma_{AB}^\dagger = \sigma_{BA}
\]
\end{definition}

In category-theoretic parlance, this is therefore a `strict symmetric monoidal functor that is \textit{identity-on-objects} and \textit{involutive}'.

\begin{definition}
  A strict dagger-symmetric monoidal category ($\dagger$-SMC) is a strict SMC with a chosen dagger functor. For a \textit{dagger-compact closed category} (with symmetric self-duality), we additionally assume $\eta_A^\dagger = \epsilon_A$.
\end{definition}

\begin{example}
For FHilb, the dagger functor sends each linear map $f : A \to B$ to its linear-algebraic adjoint $f^\dagger : B \to A$, i.e.~the unique linear map such that:
\[ 
\braket{b}{f(a)} = \braket{f^\dagger(b)}{a} 
\]
for all $a \in A, b\in B$.
\end{example}

There is a tight connection between equations between linear maps and equations between string diagrams. Namely, any equation between string diagrams involving linear maps $f, g, h, \ldots$ (and their adjoints) holds generically---i.e.~for \underline{all} linear maps $f, g, h, \ldots$---precisely when the diagrams themselves are equal. This is formally stated as a \textit{completeness} theorem:

\begin{theorem}\label{thm:Selinger}
  FHilb is complete for string diagrams.
\end{theorem}

In other words, if we don't know anything about the linear maps in a diagram (i.e.~we treat them as `black boxes'), diagrammatic reasoning is already the best we can do.

\subsection{Reference and further reading}

The quantum teleportation protocol first appeared in \cite{tele}. A diagrammatic derivation of teleportation first appeared in \cite{LE1}, and independently, also in \cite{Kauffman}.  The four variations of boxes as in (\ref{eq:smapALLFOURnew}) first appeared in \cite{SelingerCPM}. The use of caps and cups also already appeared in \cite{Penrose}. Proposition \ref{prop:closurestring} is known in quantum theory as the Choi-Jamio\l{}kowski isomorphism \cite{jamiolkowski, choi}.

The corresponding category-theoretic axiomatisation is due to Abramsky and Coecke \cite{AC1}, which was the first paper on categorical quantum mechanics, building further on Kelly's compact closed categories \cite{Kelly} by adjoining a dagger functor. Dagger compact categories had already appeared in \cite{BaezDolan} as a  special case of a more general construct in $n$-categories. Theorem \ref{thm:Selinger} is due to Selinger \cite{Selingercompleteness}. A comprehensive (at the time) survey of monoidal categories and their various graphical languages is given in~\cite{SelingerSurvey}.

\section{Quantum processes}

\begin{quote}
\em I would like to make a confession which may seem immoral: I do not believe   
absolutely in Hilbert space any more.\par \em \hfill    --- John von Neumann, letter to Garrett Birkhoff, 1935.         
\end{quote}

\noindent
Let us summarise what we have seen thus far.  In the first section we introduced a general formalism to reason about interacting processes based on diagrams.  In the following section we saw how the simple assumption that the diagrams are string diagrams allows a kindergartner to derive quantum teleportation. And then we bumped into a caveat, which we claimed came from something called \textit{causality}.

We claimed before that, rather than simply `doing' a cap effect, Aleks must perform a non-deterministic processes which might instead yield a cap-effect with some error. But that's just a bunch of words. None of the ingredients coming from string diagrams can actually help us derive this fact.

There's also a second issue here. Before Bob can correct the error, he must know which error $U_i$ happened. The only way this is possible is if Aleks picks up the phone and tells him. This requires distinguishing a phone call from a quantum system in our diagrams, i.e.~diagrams will need to involve two kinds of types: classical types and quantum types, and their distinct behaviour should also be evident in diagrammatic terms. In fact, it is the diagrammatic  formulation of quantum types that will allow us to  define this (in)famous causality postulate  which brought us here in the first place. Moreover, we will then be able to diagrammatically derive the fact that the cap-state cannot be invoked with certainty.

This section concerns specification of a  very general kind of quantum type.  Despite its generality, we will already be able to prove some highly non-trivial features of quantum systems, most notably, the no-broadcasting theorem.  This sets the stage for exploring alternative models  of quantum theory, which go beyond Hilbert spaces.

\subsection{Quantum types}

Back in Section~\ref{sec:procs-as-diagrams}, we interpreted a state meeting an effect as a probability, i.e.~a positive number:
\begin{equation}\label{eq:born-again}
  \input{./figures/state_test_paper.tikz}
\end{equation}
At that time, we didn't have a language rich enough to give a notion like `positivity', but thanks to reflection (i.e.~adjoints and conjugates) we now do. A complex number is positive precisely when it is the product of some number $\lambda$ and its conjugate $\overline{\lambda}$. Being the composition of a number and its conjugate thus gives us a generalisation of the condition of a number being positive:
\begin{equation}\label{eq:gen-positive}
  \input{./figures/gen-positive.tikz}
\end{equation}
So, we could use expressions like the one on the right to compute probabilities, and indeed this is how it is typically done in quantum theory. However, there is a way to have our cake and eat it too. Namely, we can retain the simplicity of~\eqref{eq:born-again} while secretly ensuring the result is always a positive number.

The trick is to double everything. By pairing two wires, we make a new thicker kind of wire:
\ctikzfig{doubled_wire}
which gives us a \textit{quantum} system type. We can then keep un-doubled types around for representing classical systems---something we'll use extensively in part II of this overview---obtaining this simple distinction:
\[
{\mbox{classical} \over\mbox{quantum}} = {\mbox{single wire} \over\mbox{double wires}}
\]

On these doubled wires, we can then build doubled boxes, which consist of the box itself, along with its conjugate:
\ctikzfig{double_process_1_to_1z}
When we apply this recipe to states and effects, we get a generalised positive number whenever the two meet:  
\ctikzfig{doubled_state_test}
Thus, the thing we called the generalised Born rule, applied to these new doubled processes, becomes what some will recognise as the  usual Born rule from quantum theory.

Doubling preserves string diagrams:
\[
\input{./figures/arbitrary_kleinx.tikz}\ \ \ \mapsto \ \ \ \input{./figures/double_arbitrary_klein.tikz}   
\]
as long as we are a bit careful about where the wires go for boxes with many inputs/outputs:
\ctikzfig{double_arbitrary}
In particular, this gives us the following expression for the doubled cap:
\[
\begin{tikzpicture}
	\begin{pgfonlayer}{nodelayer}
		\node [style=none] (0) at (-1, 0.5) {}; 
		\node [style=none] (1) at (1, 0.5) {};  
	\end{pgfonlayer}
	\begin{pgfonlayer}{edgelayer}
		\draw [style=boldedge, in=-90, out=-90, looseness=1.75] (0.center) to (1.center);    
	\end{pgfonlayer}
\end{tikzpicture}\ \ :=\ \input{./figures/double_cup_derivation1.tikz}\vspace{-2mm}
\]
and similarly for the cup.

It also preserves all equations between string diagrams, so, for example, our derivation of teleportation carries over to the doubled world.  What about the  converse? Do equations in  doubled diagrams carry over to their un-doubled counterparts?  The answer turns out to be no,  but luckily this is a feature, not a bug.  

\begin{proposition}\label{prop:phases}
In a theory  described by string diagrams, if:
  \beq\label{eq:phaseelim1}
\dmap{\widehat f} \ =\ \dmap{\widehat g}
  \eeq
then there exist numbers $\lambda$ and $\mu$, with  $\widehat \lambda=\widehat \mu$, such that:
\begin{equation}\label{eq:phaseelim2}
  \scalar{\lambda} \, \map{f} \ =\ \scalar{\mu} \, \map{g}
\end{equation}
The converse also holds provided $\widehat \lambda=\widehat \mu$ is cancellable. 
\end{proposition}
\begin{proof}
Let $\lambda$ and $\mu$ be:  
\[
\scalar{\lambda} \ := \ \ \input{./figures/elim3bis.tikz}\qquad \quad\qquad\scalar{\mu} \ := \ \ \input{./figures/elim4bis.tikz}  
\]
Unfolding \eqref{eq:phaseelim1} yields:
\beq\label{eq:phaseelim3}
\mapconj{f}\ \map{f} \ =\ \mapconj{g}\ \map{g} 
\eeq
So:
\[
\dscalar{\widehat \lambda} \ \ \input{./figures/elim6_paper.tikz} \ \   \dscalar{\widehat \mu}
\]
and:
\[
\scalar{\lambda} \, \map{f}\ \, \input{./figures/elim5_paper.tikz}  \ \, \scalar{\mu} \, \map{g}  
\]
\end{proof}

Typically, `cancellable' means non-zero, and the condition that $\widehat\lambda = \widehat\mu$ means the absolute values of $\lambda$ and $\mu$ coincide. In that case, the condition \eqref{eq:phaseelim2} is more commonly referred to as `equal up to a (global) complex phase $\mu\over\lambda$'. One of the uglier aspects of the standard quantum formalism is indeed that it contains redundant numbers  of the form $\mu\over\lambda$, which have no physical significance.  Thus a nice side-effect of doubling is precisely removing this redundancy.

\subsection{Pure processes and discarding}

A state $\psi$ is called \textit{normalised} if composing with its adjoint yields  1, i.e.~the empty diagram:
\[ 
\kpointbraket{\psi}{\psi} \ =\ \emptydiag   
\]
When we pass to the doubled world, we therefore have a way to throw away states $\widehat\psi$ arising from normalised states $\psi$. Simply connect the two halves together:
\[
\input{./figures/discard_prop_proof_paper.tikz}\ \  = \ 
\kpointbraket{\psi}{\psi}\ =\ \ \emptydiag
\]
This new effect, which has no counterpart in the un-doubled underlying theory,  is called \textit{discarding}, and we denote it as follows:
\[
\discard\ \, := \ \ \input{./figures/trace_def.tikz}
\]
This then immediately gives rise to new boxes  as well: 
\beq\label{eq:quantummapx}
\input{./figures/quantummapx.tikz}
\eeq
which we call \textit{impure}, or \textit{mixed}, processes---for reasons that shall become clear in the follow-up paper to this one. These impure processes naturally occur, as the diagram indicates, when we are only considering part of a composite system, and ignoring  (i.e.~discarding) another part.  In other words, we are considering systems that are in interaction with their \textit{environment}. A \textit{pure} process is then one that doesn't involve discarding, i.e.~that doesn't have  a wire connecting its two halves. Impure processes will typically be denoted by $\Phi$, and and impure states by $\boldsymbol\rho$, in contrast to the pure ones which carry a `hat'.

Discarding multiple systems is the same as discarding one, bigger system. Hence we can put any process consisting of pure processes and discarding in the form of \eqref{eq:quantummapx} by grouping all of the discarding processes together into a single effect:
\[
  \input{./figures/arbitrary_w_discardbis.tikz}
  \quad\mapsto\quad
  \input{./figures/arbitrary_w_discardbisA.tikz}
\]
As a consequence:

\begin{proposition}\label{prop:q-boxclosedness}
Any string diagram consisting of processes of the form (\ref{eq:quantummapx}) is again of that form. Hence they form a process theory.
\end{proposition}

Another consequence is that any diagram of pure processes and discarding can be put into a standard form involving a pure process and a single discard. In other words, any (possibly impure) process $\Phi$ has a \textit{purification} $\widehat f$:
\beq\label{eq:purification}
\dmap{\Phi} \ =\ \ \input{./figures/purificationI.tikz}
\eeq

Condition (\ref{EQ:daggeraxiom}) on adjoints now allows us establish an important connection between $\otimes$-separability, and purity of the \textit{reduced state} of a composite system, i.e.~what remans if we discard part of it:
\ctikzfig{bipartite_disc_rho}
  
\begin{proposition}\label{prop:reduce-pure}   
Consider a theory that admits string diagrams and in which adjoints obey \eqref{EQ:daggeraxiom}. If the reduced state of a bipartite state $\boldsymbol\rho$ is pure:  
  \begin{equation}\label{eq:rho-disc-pure} 
    \input{./figures/bipartite_disc_rho.tikz}\  \ =\ \dkpoint{\,\widehat\phi\,}      
  \end{equation}
 then it $\otimes$-separates as follows:
\[ 
 \dbistate{\boldsymbol\rho} \ =\
 \dkpoint{\,\boldsymbol\rho'}\, \dkpoint{\,\widehat\phi'\,}           
\] 
\end{proposition}
\begin{proof}
 Writing $\boldsymbol\rho$ in the form (\ref{eq:quantummapx}):
\begin{equation}\label{eq:unfold-form}
   \dbistate{\boldsymbol\rho}\, \ =\ \  \input{./figures/bistate-pfa-3.tikz} 
\end{equation}
and substituting this into (\ref{eq:rho-disc-pure}) we obtain: 
 \[
\input{./figures/bistate-disc-pfa-3.tikz} \ \  = 
\ \   \input{./figures/final1.tikz}   
\]
and hence:
\[
\input{./figures/bistate-disc-pfb-3.tikz} \ \ =\ \ \kpointconj{\phi}\kpoint{\phi}       
\]
Deforming this equation via process-state duality and transposition, we get:
\[
\input{./figures/final2NEW.tikz}
\]
Then, by (\ref{EQ:daggeraxiom}) there exist $\psi_1$ and $\psi_2$ such that:
\[
 \input{./figures/final3NEW-3.tikz} 
\]
hence:
\[
\input{./figures/final3NEW-3b.tikz}
\]
Plugging in to \eqref{eq:unfold-form} indeed yields a $\otimes$-separable state of the required form:
\[
\input{./figures/bistate-pfa-3-sep.tikz}\ \ =\ \ \input{./figures/bistate-pfa-3-sep2.tikz} 
\]
\end{proof}

By process-state duality this fact straightforwardly extend to processes:

\begin{proposition}\label{prop:discard-mix-pure-proc}   
Consider a theory  that admits string diagrams and in which adjoints obey \eqref{EQ:daggeraxiom}. If        
a  \em reduced process \em of a process $\Phi$ is pure:  
  \begin{equation}\label{eq:rho-disc-pure-proc}
\input{./figures/bekan7.tikz}\  \ =\ \dmap{\widehat f}  
  \end{equation}
 then it $\otimes$-separates as follows:
 \beq\label{eq:rho-disc-pure-proc2}
\input{./figures/bekan8.tikz}\  \ =\ \dkpoint{\,\boldsymbol\rho}\ \dmap{\widehat f}    
  \eeq
\end{proposition}
\begin{proof}
Bend the wire in (\ref{eq:rho-disc-pure-proc}):   
  \[
  \input{./figures/bekan9.tikz}\ \ = \ \ \input{./figures/bekan10.tikz}
\]
Treating the two rightmost wires above as a single system, this is the reduced state of a bipartite state. Since the reduced state is pure, by Proposition~\ref{prop:reduce-pure} it separates as follows:
\[
\input{./figures/bekan11.tikz}\ \ = \ \dkpoint{\,\boldsymbol\rho}\  \input{./figures/bekan10.tikz}
\]
Unbend the wire and we're done.
\end{proof}

\subsection{No-broadcasting}    

A \textit{cloning process} $\Delta$ is a process that takes any state as input and produces  two copies of that state as output:
\beq\label{eq:cloning}
\input{./figures/cloneboxstated.tikz} 
\eeq
It is often said that a key difference between quantum and classical processes is that the latter admits cloning. In fact, this is not entirely true if one includes probabilistic classical states into the mix, in which case there is no way to `clone' a probability distribution either.

However, what is possible classically is \textit{broadcasting}. That is, there exists a process $\Delta$  such that, when a state $\boldsymbol\rho$ is fed in and either output is discarded, we are left we $\boldsymbol\rho$:
\begin{equation}\label{eq:rho-broadcast} 
  \input{./figures/rho-broadcast-paper.tikz}
\end{equation}
It is easily seen diagrammatically that broadcasting is indeed a weaker notion than cloning:
\ctikzfig{cloneboxstatedproof-paper1}
and similarly for discarding the other output. Rather than making explicit reference to  the state  $\boldsymbol\rho$, we can also give the broadcasting equations (\ref{eq:rho-broadcast}) in a state-less (or, if you want, point-less) form:  
   \begin{equation}\label{eq:norho-broadcast}     
    \input{./figures/broadcast-paper.tikz}
  \end{equation}
  
\begin{theorem} 
 If in a theory  described by string diagrams adjoints obey (\ref{EQ:daggeraxiom}), then the theory obtained by doubling and adjoining discarding cannot have a broadcasting process.      
  \end{theorem}
\begin{proof}
By equation (\ref{eq:norho-broadcast}l) the reduced state of $\Delta$ is pure, namely a plain wire, so by Proposition \ref{prop:discard-mix-pure-proc} we have: 
  \beq\label{eq:broadcast-nopurify4}  
  \input{./figures/broadcast-nopurify4.tikz}
 \eeq 
for some state $\boldsymbol{\rho}$.  Hence it follows that:   
  \ctikzfig{broadcast-nopurify5}
  Since the identity is $\circ$-separable, so is every other process involving that type, and hence the system must be trivial for $\Delta$ to exist.
\end{proof}

\subsection{Category-theoretic counterpart}

Doubling and adding discarding has a categorical counterpart as well. Rather than building it up piecewise from doubled processes and discarding, the categorical construction just declares that all morphisms should be of the form \eqref{eq:quantummapx}:

\begin{definition}
  For a compact closed category $\mathcal C$, we can form a new compact closed category $\textrm{CPM}[\mathcal C]$ with objects $\widehat A$ for every $A \in \textrm{ob}(\mathcal C)$. The morphisms $f : \widehat A \to \widehat B$ are those morphisms $f \in \mathcal C(A \otimes A, B \otimes B)$ which are of the form:
  \ctikzfig{CPM-form}
  for some object $C$ and morphism $g : A \to C \otimes B$.
\end{definition}

One needs to do a bit of work to show that this is indeed a compact closed category. For instance, the parallel composition of $f : \widehat A \to \widehat B$ and $f' : \widehat A' \to \widehat B'$ should give something of the form:
\[ 
f \otimes f' \in \mathcal C((A \otimes A') \otimes (A \otimes A'), (B \otimes B') \otimes (B \otimes B'))
\]
which involves some reshuffling of wires. We refer to references in the next section for details.

The acronym CPM refers to `completely positive map', and indeed when you apply this construction to FHilb, you get the category whose morphisms are completely positive maps.

\subsection{Reference and further reading}  
  
The doubling construction was introduced in \cite{DeLL}, including the  proof of Proposition \ref{prop:phases}. Around the same time, the  generalisation to impure was introduced by Selinger \cite{SelingerCPM}. The idea that this can be done by adding the discarding process was put forward in \cite{SelingerAxiom}. 

The no-broadcasting theorem first appeared in \cite{Nobroadcast}, and our derivation from doubling and (\ref{EQ:daggeraxiom})   is novel. Another `generalised no-broadcasting theorem' is \cite{BBLW}, which rather than  process theories, concerns  generalised probabilistic theories.


\section{Causality}\label{sec:causality}
  
\begin{quote}
\em A new scientific truth does not triumph by convincing its opponents and making them see the light, but rather because its opponents eventually die, and a new generation grows up that is familiar with it.  
\par \em \hfill    --- Max Planck, 1936.      
\end{quote}

\noindent
The beginning of the previous century saw two revolutions in physics: quantum theory and relativity. While the first one is a  theory which associates probabilities to non-deterministic processes (usually) involving microscopic systems, the second concerns the geometry of spacetime. Evidently, since there is only one reality, those two theories should not contradict each other. Amazingly, this compatibility can already be obtained within the generality of diagrams, provided there are discarding effects. If  these discarding  effects happen to arise from doubling as described above, then many more results follow.

\subsection{Causal processes}

Suppose we apply  a process to some inputs, but then discard all of its outputs.  Then the performance has gone to waste, and we could as well  have simply discarded its inputs.  In fact, this obvious assumption is quite vital for even being able to perform science.  It allows one to `discard'  (i.e.~ignore) everything that does not directly affect an experiment,  such as for example,  stuff happening in  some other galaxy. However innocent and/or obvious this principle sounds, it has many striking consequences,  warranting an important-sounding name:

\begin{definition}\label{def:eq:causqmaps}
In a process theory where each system  has a  distinguished discarding effect, a process ${\Phi}$ is called \textit{causal} if we have: 
\beq\label{eq:causqmaps}
\begin{tikzpicture}
	\begin{pgfonlayer}{nodelayer}
		\node [style=none] (0) at (0, -1.25) {};
		\node [style=dmap] (1) at (0, 0) {${\Phi}$};
		\node [style=upground] (2) at (0, 1.5) {};
	\end{pgfonlayer}
	\begin{pgfonlayer}{edgelayer}
		\draw [style=boldedge] (1) to (0);
		\draw [style=boldedge] (1) to (2); 
	\end{pgfonlayer}
\end{tikzpicture}\ \ = \ \, \discard
\eeq
We call a theory causal if all its processes are.
\end{definition}

Note that \eqref{eq:causqmaps} should be read as discarding \underline{all} of the outputs on the LHS results in discarding \underline{all} of the inputs on the RHS. As a special case, states have no inputs,  so nothing has to be discarded in the RHS:
\[
 \begin{tikzpicture}
	\begin{pgfonlayer}{nodelayer}
		\node [style=dkpoint] (0) at (0, -0.5) {$\boldsymbol\rho$};
		\node [style=upground] (1) at (0, 0.75) {};
	\end{pgfonlayer}
	\begin{pgfonlayer}{edgelayer}
		\draw [style=boldedge] (0) to (1);
	\end{pgfonlayer}
\end{tikzpicture} \ \  =\  \emptydiag  
\]
This is the usual condition for a (possibly impure) quantum state being normalised. Similarly, effects have no outputs, so nothing has to be discarded in the LHS:
\[
\dkpointadj{\boldsymbol\rho} \ =\ \discard 
\]
Clearly, this is bad news for those who like variety:

\begin{theorem}\label{thm:only-one-effect}
In a causal theory, each system-type  admits only one effect: discarding.
\end{theorem}

In the case of two systems this means that:
\[
\discard \, \discard 
\]
is the only available effect.  So now it should be clear that, if we restrict to causal effects, teleportation is  no longer possible:
\ctikzfig{telenoclass2dub}
The way around this problem is to avoid Theorem \ref{thm:only-one-effect} simply by having a classical output, corresponding to Aleks' outgoing phone call to Bob.  As we will see in the follow-up paper \cite{CQMII}, such a classical output is what enables one to accommodate the non-determinism alluded to in Section \ref{sec:adjoints-and-unitarity}.

\subsection{Evolution and Stinespring dilation}\label{sec:stinespring}

We now consider theories arising from doubling.  In our first result, we will actually derive something that is typically assumed \textit{a priori} in the standard formulation of quantum theory \cite{vN}.

\begin{theorem}\label{thm:purecausaliso}
In  a theory obtained by doubling and adjoining discarding, the following are equivalent  for any pure process $\widehat U$:    
 \ben
 \item $\widehat U$ is causal:    
\[
\ \ \begin{tikzpicture}
	\begin{pgfonlayer}{nodelayer}
		\node [style=none] (0) at (-1.75, -1.25) {};
		\node [style=dmap] (1) at (-1.75, 0) {${\widehat U}$};
		\node [style=upground] (2) at (-1.75, 1.5) {};
	\end{pgfonlayer}
	\begin{pgfonlayer}{edgelayer}
		\draw [style=boldedge] (2) to (1);
		\draw [style=boldedge] (1) to (0.center);
	\end{pgfonlayer}
\end{tikzpicture}\, \ = \  \discard
\]
 \item $U$ is an isometry: 
\[
\input{./figures/isometry.tikz} 
\]
 \item and, $\widehat U$ is an isometry:   
\[
\input{./figures/isometrydub.tikz}
\]
\een 
\end{theorem}
\begin{proof}
Unfolding the causality equation, we have:
\[
  \input{./figures/isometryproof1.tikz}
\]
and we recover (\ref{eq:isometry}) simply  by un-bending, so $1 \Leftrightarrow 2$. $2 \Leftrightarrow 3$ follows by Proposition \ref{prop:phases}.
\end{proof}

\begin{corollary}\label{cor:purecausalisobis}
Under the assumptions of the previous theorem,  the following are equivalent for pure processes:
\ben
\item $\widehat U$ is causal and invertible, and
\item $\widehat U$ is unitary.
\een 
\end{corollary}

Dropping purity yields a standard result of quantum information theory:

\begin{theorem}[Stinespring dilation]
Under the assumptions of the previous theorem,  for every causal process $\Phi$ there exists an isometry $\widehat U$ such that:
\beq\label{eq:prestinespring}
 \dmap{\Phi} \ =\ \ \input{./figures/quantummapcaus.tikz}   
\eeq
\end{theorem}
\begin{proof} 
We have \eqref{eq:prestinespring} immediately if we let $\widehat U$ be the purification of $\Phi$, as in \eqref{eq:purification}. So, it suffices to prove that $\widehat U$ is an isometry in (\ref{eq:prestinespring}).  By causality of $\Phi$, it follows that $\widehat U$ must also be causal: 
\[
\input{./figures/quantummapcausproof-paper.tikz}
\]
which, by Theorem \ref{thm:purecausaliso} implies that $\widehat U$ is an isometry.    
\end{proof}

\subsection{No-signalling}  

In this section, we'll set doubling aside again and look at the general case of theories with discarding,  and how diagrams in such a theory relate to \textit{causal structures}. A causal structure is simply a directed graph (without cycles) that indicates some collection of events in spacetime, where an edge from $e_1$ to $e_2$ indicates that $e_1$ could have an influence on $e_2$, i.e. $e_2$ is in the \textit{causal future} of $e_1$. For example, the following causal structure:
\ctikzfig{Causalyab}
indicates that Aleks and Bob may have some shared history,  but now they have moved far away from each other, so that they can no longer directly communicate. More precisely, they are so far away that the speed of light prohibits Aleks sending any kind of message to Bob and vice versa.

A process theory is said to be \textit{non-signalling} if each process $\Phi$ in a diagram with a fixed causal structure can only have an influence on processes in the causal future of $\Phi$. We claim that  any causal process theory is non-signalling. 

First, we note that we can associate a causal structure to any circuit diagram by associating boxes to events, and declaring a box $\Phi$ to be in the causal future of another box $\Psi$ if an output of $\Psi$ is wired to an input of $\Phi$:
\ctikzfig{causal-ex-pre-paper}
Note that we can freely add some extra input/output wires to account for the fact that Aleks and Bob can interact with their own processes locally:
\ctikzfig{puricausalcomp1-pre}
However, in this situation, non-signalling dictates that, since Bob does not have access to the output of Aleks' process (as this would require Aleks sending a message faster than the speed of light), Aleks shouldn't be able to influence Bob via his input. To see this is the case, let's see things from Bob's perspective by discarding Aleks' output:
\[  
\input{./figures/puricausalcomp1-paper.tikz}
\]
and see if Bob can learn anything about Aleks' input from his own input-output pair. By causality we have:
\[
\input{./figures/rrr9.tikz}
\]
and hence it follows that:    
\[
\input{./figures/puricausalcomp2.tikz}
\]
So from Bob's perspective, his input-output pair is $\circ$-separated from Aleks' input. Thus no signalling from Aleks to Bob can take place. By symmetry it also follows that Bob cannot signal to Aleks. 

\begin{theorem}
 If a process theory has a discarding process for each type and it satisfies causality, then it is non-signalling. 
\end{theorem}

\subsection{Category-theoretic counterpart}   

In contrast to some of the previous definitions, one usually encounters the following definition in the first lesson of a course on category theory:    
  
\begin{definition} 
  An object $T$ is called \textit{terminal} if for every object $A$, there exists a unique morphism $! : A \to T$.  
\end{definition}

Its a quick exercise to show that, if a terminal object exists, it is unique (up to isomorphism). A monoidal category is called \textit{causal} if the monoidal unit $I$ is a terminal object, which is the same as saying that every morphism in the category satisfies the causality postulate.

\subsection{Reference and further reading}  
  
The causality postulate was only recently identified as a core principle of quantum theory, by Chiribella, D'Ariano and Perinotti \cite{chiri1}, as one of a series of axioms from which quantum theory was reconstructed.

Stinespring dilation first appeared within the context of C*-algebras \cite{Stinespring}. Our derivation of unitarity and Stinespring dilation from causality is new.  

The first proof of the non-signalling theorem for quantum theory can be found in \cite{Ghirardi}. The derivation of non-signalling from the causality principle  is taken from \cite{Cnonsig}. A similar result is also in \cite{FritzII}, but in those papers more structure is used in order to establish this fact.     

\section{What comes next}    

Simply by stating that a quantum type arises from doubling we derived many typical features of quantum theory.  The motivation for doubling arose from distinguishing quantum types from classical types.  However, we haven't said anything yet about classical types, nor how quantum and classical types interact. Having a grip on this is essential to understanding concepts such as mixing, measurement, and entanglement, and will enable us to give fully-comprehensive descriptions of several quantum protocols as diagrams.  This will be the content of the follow-up paper \cite{CQMII}.   

In the final paper \cite{CQMIII} we will discuss the important quantum theoretical notion of complementarity, as well as a strengthening thereof, which, among many other things, will enable us to derive quantum non-locality.  The corresponding category-theoretic  notions to classicality and complementarity  will involve certain  kinds of algebraic structures within a monoidal category, namely certain Frobenius algebras and Hopf-algebras.    

\bibliography{main}
\bibliographystyle{plain}  

\end{document}